%% file: arXiv.tex
\documentclass{article}
\usepackage[utf8]{inputenc}
\usepackage[margin=1in]{geometry}

\usepackage{booktabs} 
\usepackage[ruled]{algorithm2e} 

\SetAlFnt{\small}
\SetAlCapFnt{\small}
\SetAlCapNameFnt{\small}
\SetAlCapHSkip{0pt}
\IncMargin{-\parindent}

\usepackage[hidelinks]{hyperref}
\usepackage[numbers,square]{natbib}
\usepackage{style}

\usepackage{newunicodechar}
\newunicodechar{∗}{\ast}
\usepackage{amsmath,amsfonts,amssymb,mathtools,amsthm}
\usetikzlibrary{decorations.pathreplacing}

\hypersetup{
    colorlinks=true,
    linkcolor=darkblue,
    filecolor=magenta,      
    urlcolor=cyan,
    citecolor=darkgreen,
}

\usepackage{thmtools}
\usepackage{enumerate}
\usepackage{enumitem}

\renewcommand{\citet}[1]{\cite{#1}}

\crefname{assumption}{Assumption}{Assumptions}

\newcommand{\x}{\mathbf{x}}
\newcommand{\y}{\mathbf{y}}
\newcommand{\p}{\mathbf{p}}
\newcommand{\bu}{\mathbf{u}}
\newcommand{\bd}{\mathbf{d}}
\newcommand{\lb}{\lambda \bar{\Delta}_\ell}
\newcommand{\arc}[1]{(\overrightarrow{{#1}})}

\DeclareMathOperator*{\argmax}{arg\,max}

\DeclarePairedDelimiterX{\inp}[2]{\langle}{\rangle}{#1, #2}
\DeclarePairedDelimiter{\norm}{\lVert}{\rVert}

\newcommand{\sumi}{ \sum_{i=1}^n }
\newcommand{\sumj}{ \sum_{j=1}^m }
\newcommand{\suml}{ \sum_{\ell=1}^k }
\newcommand{\anyi}{ \forall\, i\in[n]}
\newcommand{\anyj}{ \forall\, j\in[m]}
\newcommand{\anyl}{ \forall\, \ell \in[k]}

\newcommand{\firm}{user }
\newcommand{\firms}{users }
\newcommand{\firmS}{user's }
\newcommand{\firmsS}{users' }
\newcommand{\Firm}{User }

\newcommand{\firmBworker}{user-worker }

\title{Competitive Equilibrium in Labor Economies \\ 
through the Lens of Goods and Chores Fisher Markets}

\author{%
  Bhaskar Ray Chaudhury$^{\dagger}$ \hspace{10pt}
  Christian Kroer$^{\ddagger}$ \hspace{10pt}
  Ruta Mehta$^{\dagger}$ \\[4pt]
  Tianlong Nan$^{\ddagger}$ \hspace{10pt}
  Zongjun Yang$^{\ddagger}$ %
}
\date{%
  {\small $^{\dagger}$University of Illinois Urbana-Champaign \qquad
   $^{\ddagger}$Columbia University}%
}
\begin{document}

\maketitle

\begin{abstract}
The Fisher market with goods has been studied extensively in the algorithmic game theory literature, yielding strong algorithmic results for computing competitive equilibria. In contrast, more recent work on Fisher markets with bads or chores has shown that the equilibrium set can be nonconvex, making equilibrium computation considerably more challenging.

In this paper, we study a two-sided labor market that couples the classical Fisher market with goods and the Fisher market with bads into a single unified framework. In our model, users 
demand tasks in order to derive utility, while workers supply labor {to perform} 
these tasks in exchange for earnings. Each task thus plays a dual role: it is a good for the consumption (user) side of the market and a chore for the production (worker) side. Given prices for tasks, users choose utility-maximizing bundles subject to budgets, while workers choose disutility-minimizing task bundles subject to earning requirements; the resulting choices induce demand and supply {\em endogenously} for each task, and a competitive equilibrium corresponds to prices at which these coincide. We start by showing that such markets are guaranteed to have equilibria in a very general setting with concave utilities on the goods side and convex disutilities on the chores side of the market. Moreover, we show that the optimality guarantees of the first and second welfare theorems hold for our labor market model.

We next study the computation of equilibria under linear preferences. We show that, similar to the chores setting, equilibria correspond to KKT points of an Eisenberg-Gale-like \emph{non-convex} program. Despite the non-convex characterization arising from  the chores side of the market, we go on to show a set of surprisingly positive results.  
First, we show that there exists a polynomial-time combinatorial algorithm for computing competitive equilibria in our setting, which relies on a natural Walrasian scheme for updating prices.
In the ``CEEI-like'' case where the budgets and earning requirements on each side of the market are all one, this yields a strongly polynomial-time algorithm.
Similar to recent results in the chores setting, we next show that our market admits a natural dual program, in spite of the non-convexity of both the primal and dual programs. Unlike the chores setting, we then show that the non-convex labor-market program admits a change of variables that transform it into a linear program (albeit with irrational coefficients). {Finally, leveraging this linear program, we give yet another polynomial-time algorithm for computing a competitive equilibrium while deriving an approach for addressing the irrational coefficients on the linear program in an efficient manner.} We note that, even for goods-only linear Fisher markets, obtaining such an LP formulation remains open.
\end{abstract}

\setcounter{tocdepth}{2} 

\setcounter{page}{1}

\input{intro}
\subsection{Our Contributions and Technical Highlights}
\input{sec-contributions.tex}

\input{related-work}

\input{sec-preliminaries}

\input{sec-EG}
\input{sec-combinatorial-algorithm}
\input{sec-LP-algorithm}

\section*{Acknowledgments}
The research of Bhaskar Ray Chaudhury was supported by NSF CAREER Grant CCF-2441580.
The research of Christian Kroer was supported by the Office of Naval Research awards N00014-22-1-2530 and N00014-23-1-2374, and the National Science Foundation awards IIS-2147361 and IIS-2238960.
The research of Ruta Mehta was supported by NSF Grant CCF-2334461.

\bibliographystyle{alphaurl}
\bibliography{refs}

\input{appendix}

\end{document}

%% file: intro.tex
\section{Introduction}

The classical Fisher market for goods has been at the heart of algorithmic game theory research since its inception, and by now, we have a thorough understanding of the existence, geometry, and computational complexity of a \emph{competitive equilibrium (CE)}. This involves elegant convex programs that capture CE as their KKT points~\cite{eisenberg1959consensus,eisenberg1961aggregation}, polynomial-time algorithms~\citep{ye2008path,orlin2010improved,devanur2002market}, and fast economic dynamics~\citep{wu2007proportional,cheung2013tatonnement,birnbaum2011distributed,nan2025convergence}. Interestingly, given all the algorithmic results and convex programs, the problem of determining whether a linear program (LP) can capture all CE in a Fisher market remains elusive and open~\citep{garg2013simplex,adsul2010simplex,devanur2002market}.


In contrast, Fisher markets with \emph{chores} have recently gained significant attention, where agents incur \emph{disutility} from consuming chores, but do so in order to get paid, and thereby satisfy their earning requirements. Chores markets have been shown to have much more intricate structure. 
\citet{bogomolnaia2017competitive} initiated the surge of interest in the chores markets by showing that a \emph{non-convex} program that minimizes the product of agent \emph{disutilities} still captures the set of competitive equilibria: any KKT point where every agent gets strictly positive disutility is a CE. 
However, there can be many CE in the chores setting, and the set of equilibria is disconnected. This resulting difficulty has eluded the polynomial-time computation of exact CE in the chores Fisher markets, 
while computing a CE in the exchange economy is proven to be PPAD-hard \cite{ChaudhuryGMM22}. 


In the classical chores-only Fisher market, workers \emph{supply} labor while the \emph{demand} for labor is exogenous and fixed. A natural question is to ask what happens when demand is also {\em endogenous},  determined by the preferences of agents (e.g., users or firms on labor platforms like \emph{Upwork, TaskRabbit} etc.) who derive utility from the completion of tasks. 
This paper develops such a two-sided model, where both labor supply and demand arise endogenously from agents' preferences and constraints, blending elements of both goods and chores Fisher markets.
Specifically, our model consists of $k$ divisible tasks, $n$ users, and $m$ workers. On the users’ side, tasks are viewed as goods; on the workers’ side, they are viewed as chores. Each user faces a budget constraint, while each worker has an earning requirement, mirroring the standard Fisher market formulations for goods and chores, respectively.  

This model is motivated by real-world platforms where both labor demand and supply respond to prices. Examples include freelance marketplaces such as \emph{Upwork} and \emph{Fiverr}, where clients with limited project budgets hire freelancers who decide how much to work based on their earning goals, and crowdsourcing platforms such as \emph{Amazon Mechanical Turk}, where requesters allocate budgets across microtasks and workers choose task participation to meet daily income targets. A substantial body of empirical evidence shows that gig and platform workers frequently exhibit \emph{income-targeting behavior}---they set specific daily or weekly earning goals and tend to reduce labor supply once those targets are met~\citep{camerer1997labor,farber2015you}. Modeling workers through earning requirements thus captures an empirically grounded behavioral regularity, while also introducing a natural theoretical model merging goods and chores Fisher markets.

Unlike traditional Fisher markets, our model does not assume a fixed supply or demand of tasks. Given task prices, workers select bundles of tasks that minimize disutility subject to their earning requirements, while users purchase utility-maximizing bundles of tasks subject to budget constraints. A competitive equilibrium corresponds to a price vector at which the induced labor supply for each task exactly matches firms’ demand. An obvious question is now: {do such equilibrium prices exist? If yes, what are their economic, structural, and computational properties?} 
In particular, do the equilibrium prices inherit the non-convexity and algorithmic difficulties of the chores Fisher market? Intuitively, one may expect that it does, since the disutility structure for workers is the same as in the chores Fisher market. Surprisingly, on the contrary, we obtain a series of surprising results, discussed in the next section, together with the techniques to obtain them. To the best of our knowledge, this is also the first Fisher market model to incorporate both endogenous demand and supply, which gives a rich framework to model more real-world problems. 



%% file: sec-contributions.tex
In this section, we present an overview of our main contributions. Although the high-level statements of our main theorems are not ordered exactly as they appear in the subsequent sections (\cref{sec:model-pre,sec:EG-program,sec:walrasian-algo,sec:LP}), we have chosen an order that best supports a clear exposition. We begin with a brief overview of our model; a detailed description is provided in~\cref{sec:model-pre}. The labor market consists of $n$ users, $m$ workers, and $k$ divisible tasks. Each \firm $i$ has a total budget of $B_i$ and each worker $j$ has an earning requirement of $E_j$. Further, each \firm derives linear utility from tasks, while each worker incurs linear disutility. Let $x_{i\ell}$ denote the fraction of task $\ell$ assigned to \firm $i$, and $y_{j\ell}$ the fraction assigned to worker $j$. Then, the utility of \firm $i$ is given by $u_i(\x_i)=\sum_{\ell} v_{i\ell} x_{i\ell}$, and the disutility (cost) of worker $j$ is $d_j(\y_j)=\sum_{\ell} c_{j\ell} y_{j\ell}$, where $v_{i\ell}$ and $c_{j\ell}$ represent the per-unit utility and disutility of task $\ell$ for \firm $i$ and worker $j$, respectively. \textcolor{black}{ We remark that our results on the existence of CE, and the welfare theorems extend to more general agent preferences, including concave utilities for firms and convex disutilities for workers. However, we focus on linear preferences for our computational and structural results on CE, in line with the dominant body of Fisher market literature that has developed around linear valuations over the past several decades~\cite{devanur2002market, orlin2010improved, cole2017convex, nan2025convergence}.}


\paragraph{Competitive Equilibrium in the Labor Market.} Fixing a price vector $\p = (p_1,$ $ p_2, \dots, p_k)$ for the tasks induces a natural notion of \emph{demand} and \emph{supply} for each task. The \emph{demand} for task $\ell$ is defined as the total number of work hours that the users are willing to sponsor when spending their budgets optimally. Formally, the demand for task $\ell$ is given by $\sum_{i} x_{i\ell}$, {for $\x_i \in \mathcal{D}^*_{i}(\p), \forall i$, where $\mathcal{D}^*_{i}(\p)$}
is the set of all optimal allocations for \firm $i$, obtained by maximizing its utility $u_i({\x}_i)$ subject to the budget constraint $\sum_{\ell} x_{i\ell} \cdot p_\ell \leq B_i$.

Similarly, the \emph{supply} for task $\ell$ is the total number of work hours that workers are willing to provide when choosing tasks that meet their earning requirements optimally. The supply is induced by $\mathcal{S}_j^*(\p)$, which is the set of allocations that minimize worker~$j$’s disutility $d_j({\y}_j)$, subject to the constraint $\sum_{\ell} y_{j\ell} \cdot p_\ell \geq E_j$. A price vector $\p$ is said to constitute a CE if there exists allocations $\x,\y$ such that each user $i$ gets $\x_i$ from their demand, each worker $j$ gets $\y_j$ from their supply, and the market is cleared with no excess or scarcity on every task, i.e., $\sum_i x_{i\ell} = \sum_j y_{j\ell}, \;\forall\, \ell \in [k]$.

\begin{example}
\label{example: CE}
There are $n=2$ \firms (each with unit budget), $m=2$ workers (each with unit earning requirement), and $k=3$ tasks. For the three tasks, \firm $1$ has unit utility $(2, 1,0)$, \firm $2$ has unit utility $(0, 1, 2)$, worker $1$ has unit disutility $(1, 2025, 1)$, and worker $2$ has unit disutility $(2025, 1, 2025)$. Then, as shown in \Cref{fig: example}, $\p = (\$1, \$0.5, \$1)$ is an equilibrium price, with an agent-wise allocation that clears the market. $\p = (\$1, \$1, \$1)$ is a non-equilibrium price, under which the demand and the supply cannot match.
\end{example}
\input{picture}

{{\em Equilibrium Existence and Welfare Theorems.}} Observe that unlike classical Fisher markets for goods or chores, the supply in our setting is not fixed--it is endogenous, and depends on the prices of the tasks. This introduces a non-trivial but economically interesting nuance to the labor market model. Nevertheless, in~\cref{app:prop-existence}, we show that a CE always exists using {a carefully designed} fixed-point argument. Our proof accommodates more general preferences: specifically, we establish existence when \firms have concave utility functions and workers have convex disutility functions. We further show that the equilibria in our labor market setting enjoy several key welfare properties, analogous to those in classical goods-only and chores-only economies. In particular, they satisfy both the First and Second Welfare Theorems, guaranteeing Pareto efficiency and the ability to support any Pareto-efficient allocation as a CE under suitable budgets and earning requirements. 

\begin{theorem*} {(Informal)}
A CE exists whenever all \firms have concave utility functions and all workers have convex disutility functions. Furthermore, our CE satisfies the first and the second welfare theorems.
\end{theorem*}

\paragraph{The Eisenberg Gale (Non-Convex) Characterization.} In Fisher markets for goods and chores, many structural and computational properties of competitive equilibria are derived via a \emph{non-linear optimization program} defined over the allocation space. In the case of goods, this is commonly known as the Eisenberg-Gale (EG) program~\citep{eisenberg1959consensus}, which maximizes the weighted sum of logarithmic utilities across all feasible allocations. The price $p_{\ell}$ for each good $\ell$, is the dual variable associated with the allocation constraint for it ($\sum_{i} x_{i\ell} = 1$) in the EG program. Observe that the EG program is convex, and this yields several desirable structural and computational properties for CE in the goods Fisher markets. In particular, it has been used to (i) show uniqueness of equilibrium prices, and (ii) design polynomial time algorithms for computing a CE.  

For the case of chores,~\citet{bogomolnaia2017competitive} introduces an EG-like program that characterizes all CEs as the Karush-Kuhn-Tucker (KKT) points\footnote{Formally, a KKT point of an optimization problem refers to a primal point that, together with a corresponding dual point, satisfies the KKT conditions.} of a program that minimizes the weighted sum of logarithmic disutilities across feasible allocations. However, in contrast to the goods setting, this program is {non-convex}, as it minimizes a strongly concave function. As a result, the structure of equilibria differs significantly: CE in the chores setting may admit multiple market-clearing price vectors, and the set of equilibrium prices can be non-convex, or even disconnected!
\citet{bogomolnaia2017competitive} further suspect that this disconnectedness could lead to inherent computational challenges. To date, no polynomial-time algorithm is known for computing a{n exact} CE in the chores market.

Given the foregoing characterizations of CE in the goods and chores markets, a natural question is whether a similar characterization exists in the labor market. Indeed, though perhaps unsurprisingly, we find that all CE in the labor market can also be characterized as the KKT points of a program defined over the allocation space. Specifically, this program maximizes the difference between the weighted sum of logarithmic utilities of the \firms and the weighted sum of logarithmic disutilities of the workers (see \Cref{sec:EG-program}). However, similar to the chores case, the resulting objective is no longer concave, making the program \emph{non-convex}. \emph{Therefore, as a first instinct, (i) one would expect the structure of CE to be non-convex-- potentially involving multiple, disconnected sets of equilibrium prices, and (ii) not anticipate the existence of a polynomial-time algorithm for computing a CE in this setting. However, in what follows, we uncover some remarkable structure in  labor markets, which will eventually give us surprising results on the geometry and computational complexity of the CE.}

\paragraph{{Poly-time Algorithm for Exact CE: Structure of Excess Demand.}}
\textcolor{black}{Many polynomial-time algorithms for computing competitive equilibria in Fisher markets with goods rely fundamentally on the Eisenberg–Gale convex program. Unfortunately, such a convex-program formulation is does not exist in our setting. Consequently, we build our algorithm, following an alternative paradigm: a simple, intuitive, yet powerful price-adjustment principle developed by Leon Walras~\citep{walras1900elements}.} 

The core idea is as follows: (i) for any fixed price vector, compute an optimal allocation based on buyer preferences and spending constraints, and (ii) increase the prices of those goods where the total aggregate demand from the buyers exceeds the supply. It can be shown that with a careful implementation, this principle guarantees convergence to a CE in the Fisher market with goods {\cite{devanur2002market,duan2016improved}}. 
However, such a procedure (and its adaptations) is not known to converge in the Fisher market with chores. In fact, \citet{ChaudhuryGMM22, ChaudhuryKMNtatonnement25} shows that sometimes Walrasian updates cause temporary divergence from a CE in the chores setting. \emph{Despite such divergent properties of Walrasian price-adjustment algorithms in the chores economy, we show a surprising convergence result in the labor market!}


{In particular, we carefully design a Walrasian price-adjustment process for the labor market, such that the overall interaction between the goods side (users demanding tasks) and chores side (workers supplying tasks) leads to {\em well-behaved} dynamics. We show that the demand structure induced by the users mitigates the instability induced by the chores-side of the market, and by leveraging on a novel potential function and structural properties unique to the labor market we show that the price-adjustment process converges to CE in polynomial time. Below is a more detailed yet intuitive explanation.}



For a given price vector $\p$ and corresponding optimal allocations $\x, \y$, the \emph{surplus}, or \emph{excess demand} for task~$\ell$ is defined as $\Delta_\ell= p_\ell \cdot \left( \sum_i x_{i\ell}- \sum_j y_{j\ell}\right)$, 
which represents the value of the net demand (demand minus supply) for task~$\ell$ at price $\p$. In the Fisher market for goods (or chores), the notion of excess demand can be defined analogously. For goods, it is given by $p_\ell \cdot \left( \sum_i x_{i\ell} - 1 \right)$, and for chores, it takes the form $p_\ell \cdot \left( 1 - \sum_j y_{j\ell} \right)$. 

In the goods-only setting, the Walrasian principle prescribes increasing the price of the good with the largest positive excess demand, i.e., the item $\ell$ with the highest value of $\Delta_{\ell} = p_\ell \cdot \left( \sum_i x_{i\ell} - 1 \right)$. The intuition is to make the over-demanded good relatively more expensive, thereby incentivizing agents to reallocate their spending toward other goods. This price adjustment naturally reduces excess demand. Let $b_{i\ell} = p_{\ell} x_{i \ell}$ 
denote the amount of money agent $i$ spends on good~$\ell$. Then $\Delta_{\ell}$ can be expressed as $\sum_i b_{i\ell} - p_\ell$. As $p_\ell$ increases and agents shift spending away from the now pricier good, each $b_{i\ell}$ is non-increasing. Consequently, the excess demand $\Delta_\ell$ declines.

Let us examine this phenomenon in the chores-only setting. The Walrasian principle applied to chores requires us to increase the price of the chore with the highest excess demand, i.e., task $\ell$ wth maximum value of $\Delta_{\ell} =  p_{\ell} \cdot (1- \sum_j y_{j \ell})$. Since agents incur disutility from consuming chores and have earning requirements, increasing the price is the only way to make the least demanded chores more attractive. Define the money earned by agent $j$ from chore $\ell$ by $e_{j \ell} = p_{\ell} y_{j \ell}$, and observe that $\Delta_{\ell} =  p_{\ell} - \sum_j e_{j \ell}$. When we increase $p_{\ell}$, crucially observe that $e_{j \ell}$ is non-decreasing, but $p_{\ell}$ is strictly increasing. In certain cases, the increase in $p_\ell$ may dominate the increase in $\sum_j e_{j\ell}$, causing $\Delta_\ell$ to \emph{increase} rather than decrease. Therefore, one cannot argue for monotone decrease of excess demand with a Walrasian procedure in the chores economy!

Finally, we observe the Walrasian updates in the labor market. The principle requires us to increase the prices of the tasks with the highest excess demand $\Delta_{\ell} = p_{\ell} (\sum_i x_{i \ell} - \sum_j y_{j \ell})$. These are tasks for which \firms are willing to demand more hours than workers are willing to provide. Naturally, increasing the price of such tasks reduces their attractiveness to \firms while simultaneously making them more appealing to workers. Now observe that $\Delta_{\ell}$ can be written as $\Delta_{\ell} = \sum_i b_{i \ell} - \sum_j e_{j \ell}$, where $b_{i \ell} = p_{\ell} x_{i \ell}$ and $e_{j \ell} = p_{\ell} y_{j \ell}$. Further, note that $b_{i \ell}$ is inverse monotone in $\p$, and $e_{j \ell}$ is monotone in $p_{\ell}$, making $\Delta_{\ell}$ also inverse monotone in $p_{\ell}$. Therefore increasing $p_{\ell}$ will not increase $\Delta_{\ell}$, but increasing $p_{\ell}$ substantially may help us decrease $\Delta_{\ell}$. \emph{The crucial structural distinction from the chores-only setting lies in the \emph{elasticity of supply}. In the chores economy, supply is fixed, and the leading term in excess demand is simply \( p_\ell \), which increases with price rise. In contrast, in the labor market, that term is replaced by the adaptive \firm bids \( \sum_i b_{i\ell} \), which decreases with price rise. This key difference enables us to design an effective Walrasian dynamics in the labor market, even though they fail in the chores-only setting.} {We manage to show polynomial-time convergence of this dynamics to an exact CE.} 
(See~\cref{sec:walrasian-algo} for a detailed explanation). 

\begin{theorem*}
    There exists a combinatorial algorithm that determines a CE in a linear labor market in polynomial time. Further, the number of arithmetic operations performed by our algorithm is polynomial in $n,m,k, \log(\max_i B_i), \log(\max_j E_j)$.
\end{theorem*}



Our algorithm is fully combinatorial and does not rely on continuous optimization oracles as black boxes. Moreover, its total arithmetic complexity is independent of the utility and disutility values, implying that it runs in strongly polynomial time when all budgets and earnings are equal. We refer to this special case as the Competitive Equilibrium with Equal Earnings and Incomes (CEEEI), drawing parallels to the classical CEEI model in goods-only markets and the CEEE model in chores-only settings.

\begin{corollary*}
    There exists a strongly polynomial-time algorithm to determine a CEEEI when agents have linear valuation functions.
\end{corollary*}

We further argue that CEEEI is well-motivated from both fairness and efficiency perspectives. For instance, consider a disaster management scenario where a group of laborers is assigned to complete tasks for various charities. In such a context, a fair and Pareto-efficient assignment—balancing the interests of both laborers and charities-- is highly desirable. CEEEI captures these goals through an equilibrium framework (similar to CEEI and CEEE).

\paragraph{{Geometry of the CE set.}} Given the strong computational guarantees established for our labor market, especially in sharp contrast to the chores-only setting-- it is natural to ask whether similarly strong structural results hold for the geometry of competitive equilibria. Unlike goods-only Fisher markets, our model does not admit a convex program formulation in the allocation space.

\textcolor{black}{Nevertheless, we make a series of nontrivial structural observations showing that the set of CE price vectors is in fact convex (\Cref{thm: convexity-of-prices}). The key insight underlying these results is that the presence of a goods-like side of the market—namely, firms with utility-maximizing demand—rules out the pathological price behaviors that arise in chores-only economies.}

\textcolor{black}{Intuitively, while the chores-only Fisher market can admit multiple, highly unstructured equilibrium price vectors, the demand side in our two-sided model disciplines prices-- any price vector that excessively inflates the price of a task may still clear on the supply side if it were a chores-only market, but it will fail to generate demand from firms and hence cannot constitute an equilibrium in the labor market. In this way, the goods side anchors prices and restores geometric regularity to the equilibrium set.}


\begin{theorem*}
    The set of CE prices in a linear labor market is convex. 
\end{theorem*}

The convexity of the CE price set strongly suggests the existence of a convex program characterizing equilibria.  
Since the EG-like formulation discussed above is clearly non-convex, the next obvious choice is to work in the dual space of prices as done in~\citep{cole2017convex,chaudhury2024competitive}. This will also help prove structural properties like the convexity of the CE price set. Next we show that, in fact a {\em linear-programming} formulation is possible. We note that, this is surprising because the lack of or the need for an LP formulation even
for the goods-only Fisher market is well-known~\citep{devanur2002market,garg2013simplex,adsul2010simplex}. 


\paragraph{More Structured than CE in the Goods Market? An LP Characterization of CE Prices.} For both the goods-only~\citet{cole2017convex} and chores-only~\citet{chaudhury2024competitive} settings, there exist well-known (non-)convex programs formulated in the price space that characterize competitive equilibria (CE) as their KKT points. These programs can be viewed as natural duals to the respective EG-like formulations and typically involve prices and variables capturing each agent’s \emph{inverse bang-per-buck}, i.e., the inverse of the maximum utility-to-price ratio (referred to as $\beta_i$) in the goods setting, and the inverse of the minimum disutility-to-price ratio in the chores setting (referred to as $\alpha_j$). Building on this framework, we construct an analogous and natural EG-dual program for the labor market. 

\noindent
\begin{minipage}[t]{0.308\textwidth}
\[
\begin{array}{ll}
\vspace{6pt}
&\text{Goods Market EG-Dual}\\
& \max\; \sum_i B_i \log \beta_i - \sum_{\ell} p_{\ell} \\
& \;\; \text{s.t.} \;\; \hspace{4.2pt} p_{\ell} \geq \beta_i v_{i \ell}  \quad \forall\, i, \ell \\
& \hspace{27.2pt} \beta_i \geq 0 \hspace{8.7pt} ~~\quad \forall\, i \\
& \hspace{27.2pt} p_{\ell} \geq 0 \hspace{9pt} ~~\quad \forall\, \ell  
\end{array}
\]
\end{minipage}
\hfill
\begin{minipage}[t]{0.308\textwidth}
\[
\begin{array}{ll}
\vspace{6pt}
&\text{Chores Market EG-Dual}\\
& \min\; \sum_j E_j \log{\alpha_j} - \sum_{\ell} p_{\ell} \\
& \;\; \text{s.t.} \;\; \hspace{4.2pt} p_{\ell} \leq \alpha_j c_{j \ell} \quad \forall\, j, \ell  \\
& \hspace{26.5pt} \alpha_j \geq 0 \hspace{21pt} ~~\forall\, j\\
& \hspace{27.5pt} p_{\ell} \geq 0 \hspace{22pt} ~~\forall\, \ell \\
& \hspace{23.5pt} \sum_{\ell} p_{\ell} = \sum_j E_j\\
\end{array}
\]
\end{minipage}
\hfill
\begin{minipage}[t]{0.37\textwidth}
\[
\begin{array}{ll}
\vspace{5pt}
& \hspace{10pt} \text{Labor Market EG-Dual} \\
& \max\; \sum_i B_i \log {\beta_i} - \sum_j E_j \log{\alpha_j} \\
& \;\;\;\text{s.t.} \;\; \hspace{4.2pt} p_{\ell} \geq \beta_i v_{i \ell} \, \quad \forall\, i, \ell \\
& \hspace{30pt} p_{\ell} \leq \alpha_j c_{j \ell} \quad \forall\, j, \ell \\
& \hspace{30pt} p_{\ell} \geq 0 \hspace{22.5pt} ~~\forall\, \ell \\
& \hspace{30pt} \beta_i \geq 0 \hspace{22.5pt} ~~\forall\, i \\
& \hspace{30pt} \alpha_j \geq 0 \hspace{21.5pt} ~~\forall\, j 
\end{array}
\]
\vspace{1pt}
\end{minipage}

We show that the KKT points of the {\em Labor Market EG-Dual} correspond to CE in the labor market. Each variable $x_{i\ell}$ can be interpreted as the dual variable associated with the constraint $p_\ell \geq \beta_i v_{i\ell}$, while each $y_{j\ell}$ corresponds to the dual of the constraint $p_\ell \leq \alpha_j c_{j\ell}$. {The market-clearing conditions can be derived from the KKT conditions combining both primal and dual variables}. Although the Labor Market EG-Dual is not convex in its original form, a simple change of variables transforms it into a linear program (LP)! Specifically, we take logarithms of the key constraints to obtain:
\begin{align*}
    \log p_{\ell} \geq \log \beta_i + \log v_{i \ell}, \text{ and } \log p_{\ell} \leq \log \alpha_j + \log c_{j \ell}. 
\end{align*}
Letting $\tilde{p}_\ell = \log p_\ell$, $\tilde{\beta}_i = \log \beta_i$, and $\tilde{\alpha}_j = \log \alpha_j$, we arrive at a linear formulation. In these new variables, every CE corresponds to an optimal solution of the following LP: 
\begin{equation}
    \begin{aligned}
        \max \quad & \sum\nolimits_i B_i \tilde{\beta}_i - \sum\nolimits_j E_j \tilde{\alpha}_j \\ 
        \textnormal{s.t.} \quad & \tilde{p}_{\ell} \geq \tilde{\beta}_i + \log(v_{i \ell})  \quad \hspace{2pt} \forall\, i, \ell \\ 
        & \tilde{p}_{\ell} \leq \tilde{\alpha}_j + \log(c_{j \ell}) \quad \forall\, j, \ell. 
    \end{aligned}
    \label{program:EG-Dual-log}
\end{equation}

This transformation is made possible by a key structural feature of the Labor Market EG-Dual: the absence of a price term in the objective. Unlike the duals for goods or chores (which include terms like $\sum_\ell p_\ell$ in the objective), the labor market EG-Dual exhibits a kind of \emph{log-linear homogeneity} that enables this exact linearization.  We note further that an important open problem is whether CE in the goods Fisher market admit a direct LP formulation \cite{garg2013simplex,adsul2010simplex}. In this context, the existence of a clean LP characterization for CE in labor markets is both surprising and noteworthy.

However, we note that the LP formulation may involve irrational inputs, specifically the terms $\log(v_{i\ell})$ and $\log(c_{j\ell})$. This raises an immediate concern about whether the LP can be solved exactly in polynomial time. To address this, we analyze the dual of the LP and propose solving a truncated version, where the irrational inputs (coefficients of the objective function) are approximated to polynomially many bits of precision. We then show that, despite this truncation, it is possible to recover an exact CE via a non-trivial, yet polynomial-time post-processing procedure.

\begin{theorem*}
    There exists a linear program whose optimal solutions correspond to a CE in the labor market. This LP has irrational coefficients, but can be solved in time polynomial in $n,m,k, \log(\max_i B_i), $ $\log(\max_j E_j), \log(\max_{i,\ell} v_{i \ell}), \log(\max_{j, \ell} c_{j \ell})$.
\end{theorem*}

\paragraph{Model Based on Alternate Behavior of the Workers.} For completeness, we also discuss in Appendix~\ref{app-alternate-model} an alternate formulation where users behave identically, but workers are modeled as profit maximizers rather than disutility minimizers. We are able to derive similar structural and computational results on the equilibrium by adapting the Eisenberg–Gale (EG) convex program. Consequently, while this formulation offers a natural parallel to our main model, its analysis is relatively straightforward. However, it lacks the empirically grounded behavioral motivation incorporated by the income-targeting formulation.

\subsection{Summary and Discussion.} To summarize, bringing together the demand and the supply side in the Fisher market setting, we introduce a model for competitive pricing in labor markets.
Under very general utility and cost functions, we establish the existence of competitive equilibrium (CE) and prove the fundamental first and second welfare theorems.

Computationally and structurally, it is natural to expect such a market to be harder than the goods-only and chores-only market. We show to the contrary in a series of suprising results: First, we design a combinatorial polynomial-time algorithm to compute a CE, which becomes strongly polynomial in the special case of equal incomes and equal earnings. Secondly, despite incorporating a chores-like component, the labor market model exhibits very well-behaved structure: the set of CE prices is convex, and the equilibrium prices can be characterized by a linear program. Notably, this LP formulation remains efficiently solvable even though it may involve irrational coefficients. We believe that our work can serve as an initiation to several research directions-- opening up new questions for classical theories and directions for emerging areas. We highlight two of them here:
\begin{itemize}
    \item \emph{New Problems in Fair Division:} The special cases of equal incomes in goods Fisher markets and equal earnings in chores Fisher markets have received significant attention in the fair division literature, owing to their strong fairness and efficiency guarantees. These settings yield allocations that are envy-free and Pareto-efficient. Building on this, several relaxations have been proposed and studied in detail for indivisible goods-- such as Almost CEEI~\citep{segal2020competitive}, EF1 + PO~\citep{aziz2022fair},-- to preserve fairness and efficiency in discrete settings. 
    
    We believe that exploring analogous questions for indivisible tasks in labor markets is a compelling direction for future research. In particular, we show that a CEEEI (Competitive Equilibrium with Equal Earnings and Incomes) offers a promising benchmark: it can guarantee envy-freeness among both \firms and workers as well as Pareto-efficiency in the allocation. Investigating whether such guarantees, or appropriate relaxations of the same, can be achieved in discrete labor markets-- along with tractable algorithms, presents a rich and natural avenue for future work.  
    
    \item \emph{Investigating Natural Dynamics in the Labor Market:} Dynamics have been part of equilibrium computation since its inception. Several dynamics such as \emph{Tâtonnement}~\citet{nan2025convergence}, \emph{Pacing}~\citet{gao2021online,yang2024online}, or \emph{Proportional Response}~\citet{zhang2011proportional}, are known to converge to a CE in the goods Fisher market. The convergence results are unknown for the chores Fisher market, primarily due to the irregular geometry of the CE in the chores market. In contrast, the CE in labor markets is well-behaved, and a Walrasian update converges to it in polynomial time. These features suggest that natural dynamics may exhibit convergence properties in this setting. We believe that investigating the convergence and also the convergence rates of such dynamics in labor markets represents an interesting future research direction. 

\end{itemize}

%% file: picture.tex
\begin{figure}[htbp]
\centering
\begin{minipage}[t]{0.48\linewidth}
\centering
\begin{tikzpicture}[scale=0.1]
\tikzstyle{every node}+=[inner sep=0pt]
\draw [black] (11.4,-22.9) circle (3);
\draw (11.4,-22.9) node {\small$\#1$};
\draw [black] (69.1,-14.9) circle (3);
\draw (69.1,-14.9) node {\small$\#1$};
\draw [black] (39.7,-31.1) circle (3);
\draw (39.7,-31.1) node {\tiny{${\$0.5}$}};
\draw [black] (39.7,-31.1) circle (2.4);
\draw [black] (39.7,-14.9) circle (3);
\draw (39.7,-14.9) node {\small$\$1$};
\draw [black] (39.7,-14.9) circle (2.4);
\draw [black] (70.1,-46.4) circle (3);
\draw (70.1,-46.4) node {\small$\#2$};
\draw [black] (11.4,-40.1) circle (3);
\draw (11.4,-40.1) node {\small$\#2$};
\draw [black] (39.7,-46.4) circle (3);
\draw (39.7,-46.4) node {\small$\$1$};
\draw [black] (39.7,-46.4) circle (2.4);
\draw [black] (14.28,-23.73) -- (36.82,-30.27);
\fill [black] (36.82,-30.27) -- (36.19,-29.56) -- (35.91,-30.52);
\draw (24.73,-27.55) node [below] {$1$};
\draw [black] (14.29,-22.08) -- (36.81,-15.72);
\fill [black] (36.81,-15.72) -- (35.91,-15.45) -- (36.18,-16.41);
\draw (26.35,-19.46) node [below] {$0.5$};
\draw [black] (14.33,-40.75) -- (36.77,-45.75);
\fill [black] (36.77,-45.75) -- (36.1,-45.09) -- (35.88,-46.06);
\draw (24.91,-43.83) node [below] {$0.5$};
\draw [black] (14.26,-39.19) -- (36.84,-32.01);
\fill [black] (36.84,-32.01) -- (35.93,-31.78) -- (36.23,-32.73);
\draw (26.42,-36.14) node [below] {$1$};
\draw [black] (42.38,-32.45) -- (67.42,-45.05);
\fill [black] (67.42,-45.05) -- (66.93,-44.25) -- (66.48,-45.14);
\draw (53.91,-39.25) node [below] {$2$};
\draw [black] (42.7,-14.9) -- (66.1,-14.9);
\fill [black] (66.1,-14.9) -- (65.3,-14.4) -- (65.3,-15.4);
\draw (54.4,-15.4) node [below] {$0.5$};
\draw [black] (41.75,-44.21) -- (67.05,-17.09);
\fill [black] (67.05,-17.09) -- (66.14,-17.34) -- (66.87,-18.02);
\draw (54.93,-32.12) node [right] {$0.5$};
\end{tikzpicture}
\end{minipage}%
\hfill
\begin{minipage}[t]{0.48\linewidth}
\centering
\begin{tikzpicture}[scale=0.1]
\tikzstyle{every node}+=[inner sep=0pt]
\draw [black] (11.4,-22.9) circle (3);
\draw (11.4,-22.9) node {\small$\#1$};
\draw [black] (69.1,-14.9) circle (3);
\draw (69.1,-14.9) node {\small$\#1$};
\draw [black] (39.7,-31.1) circle (3);
\draw (39.7,-31.1) node {\small$\$1$};
\draw [black] (39.7,-31.1) circle (2.4);
\draw [black] (39.7,-14.9) circle (3);
\draw (39.7,-14.9) node {\small$\$1$};
\draw [black] (39.7,-14.9) circle (2.4);
\draw [black] (69.1,-46.4) circle (3);
\draw (69.1,-46.4) node {\small$\#2$};
\draw [black] (11.4,-40.1) circle (3);
\draw (11.4,-40.1) node {\small$\#2$};
\draw [black] (39.7,-46.4) circle (3);
\draw (39.7,-46.4) node {\small$\$1$};
\draw [black] (39.7,-46.4) circle (2.4);
\draw [black] (14.29,-22.08) -- (36.81,-15.72);
\fill [black] (36.81,-15.72) -- (35.91,-15.45) -- (36.18,-16.41);
\draw (26.35,-19.46) node [below] {$1$};
\draw [black] (14.33,-40.75) -- (36.77,-45.75);
\fill [black] (36.77,-45.75) -- (36.1,-45.09) -- (35.88,-46.06);
\draw (24.91,-43.83) node [below] {$1$};
\draw [black] (42.36,-32.48) -- (66.44,-45.02);
\fill [black] (66.44,-45.02) -- (65.96,-44.2) -- (65.5,-45.09);
\draw (53.41,-39.25) node [below] {$1$};
\draw [black] (42.7,-14.9) -- (66.1,-14.9);
\fill [black] (66.1,-14.9) -- (65.3,-14.4) -- (65.3,-15.4);
\draw (54.4,-15.4) node [below] {$x$};
\draw [black] (41.75,-44.21) -- (67.05,-17.09);
\fill [black] (67.05,-17.09) -- (66.14,-17.34) -- (66.87,-18.02);
\draw (54.93,-32.12) node [right] {$1-x$};
\end{tikzpicture}
\end{minipage}

\caption{An equilibrium \textit{(left)} and a non-equilibrium \textit{(right)} for the instance in \Cref{example: CE}, with $2$ \firms \textit{(left nodes)}, $3$ tasks \textit{(middle nodes)}, and $2$ workers \textit{(right nodes)}. The demand and supply are shown with directed edges, with the corresponding amounts of allocation marked on the edges.}
\label{fig: example}
\end{figure}
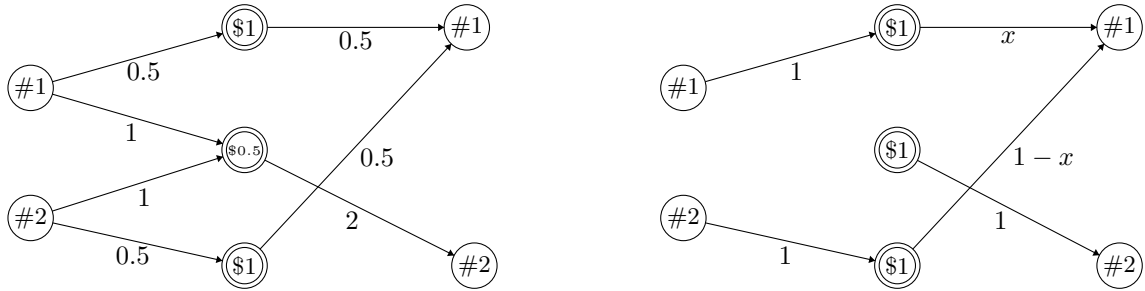

%% file: related-work.tex
\subsection{Related Work}

\paragraph{CE with goods.}
Algorithms for computing a CE have been developed since the 1960s, 
and gained significant attention in the computer science community beginning in the 2000s~\citep{codenotti2004computation}. 
In goods Fisher markets, 
the EG program plays a central role in the design of polynomial-time algorithms, 
as the ellipsoid method yields the first polynomial-time algorithm in the literature. 
\citet{ye2008path} showed a polynomial-time algorithm based on an interior point method, which significantly improved the complexity bound. 
More recently, the EG program has also served as the foundation for designing algorithms for solving large-scale market equilibria~\citep{gao2020first,nan2023fast,liu2025pdhcg}. 

\citet{devanur2002market} provided the first combinatorial algorithm to find an exact CE. Their algorithm views the allocation of money as flows and iteratively improves the balanced flow. 
\citet{orlin2010improved} improved the running time and presented the first strongly polynomial-time algorithm. 
There is also a line of research that studies tâtonnement as a polynomial-time algorithm for computing approximate CE; see, for example, \citet{codenotti2005market, cheung2013tatonnement, nan2025convergence}. 

Another powerful convex program for Fisher market with goods is known as Shmyrev program~\citep{shmyrev2009algorithm} which is defined over the spending money flow space. 
\citet{birnbaum2011distributed} showed that the well-known proportional response dynamics~\citep{wu2007proportional} is equivalent to applying the mirror descent method to the Shmyrev program. 
\citet{cole2017convex} showed a formal connection between the EG and Shmyrev programs. 

In more general goods exchange markets, 
polynomial and strongly polynomial algorithms have also been extensively studied in the literature~\citep{jain2007polynomial,ye2008path,duan2015combinatorial,garg2019strongly,chaudhury2018combinatorial,duan2016improved}, 
along with new convex programs~\citep{devanur2016rational}. 

\paragraph{CE with chores.} 
The EG-like characterization introduced in~\citet{bogomolnaia2017competitive} initialized the study on CE with chores in Fisher and exchange markets. 
Computing a CE with chores has been seen as a hard problem, due to its associated nonconvexity and disconnectedness:~\citet{chaudhury2024competitive} showed that computing CE in exchange markets with chores is PPAD-hard. 
To compute an exact CE, polynomial-time algorithms are only known when the number of agents or items is fixed~\citep{branzei2019algorithms,garg2020computing}. 

\citet{chaudhury2024competitive} proposed a combinatorial polynomial-time algorithm to compute an approximate CE. 
By using an exterior point method, 
\citet{boodaghians2022polynomial} showed a polynomial-time algorithm to compute approximate CE based on the EG-like program. 
\citet{chaudhury2024competitive} introduced a new EG dual program that addresses the issue of the open constraint in the feasible set, thereby opening up the possibility of applying standard optimization methods to compute CE with chores. 
They proposed a greedy Frank–Wolfe algorithm, achieving improved time complexity. \citet{chen2024computing} applied the difference-of-convex algorithm to a variant of the EG dual program, and showed a linear convergence rate by considering an error bound condition. 

\paragraph{Related market models.} 
Several labor market models have been studied in the economics and computer science literature. 
One line of work studies many-to-one job matching markets; 
see, for example, \citet{kelso1982job,bando2012many}. 
In contrast to our setting, 
these models do not involve tasks and focus solely on “direct” matching between \firms and workers. 
Very recently, \citet{garg2024matching} introduced a one-sided matching market with goods and chores. 
However, their model only considers matching between a single group of agents and items, making it fundamentally different from ours. 
\citet{singer2011pricing} studied an online labor market in which a single requester assigns tasks to potential workers through a pricing mechanism. 
In contrast, our model considers task allocation on both \firm and worker sides. 



%% file: sec-preliminaries.tex
\section{Model and Preliminaries}
\label{sec:model-pre}

We describe our labor market model. We assume there are $n$ \firms, $m$ workers, and $k$ divisible tasks that \firms can request workers to work on; denote the set of users, workers, and tasks as $[n], [m], [k]$, respectively. Each \firm $i$ has a budget $B_i>0$ to spend to get its tasks done, and while each worker $j$ wants to earn $E_j>0$ by doing the tasks. Naturally, the supply and demand of these tasks are governed by the prices of the tasks and the preferences of the \firms and workers.  

For each unit of task $\ell \in [k]$, \firm $i\in[n]$ receives {\em utility} $u_{i\ell}$, while worker $j$ receives {\em disutility} $c_{i\ell}$. Their valuation across tasks is linear (unless stated otherwise): let $x_i = (x_{i1}, \cdots, x_{ik}) \in \mathbb{R}_+^k$ denote the allocation of the tasks to \firm $i$, where $x_{i\ell}$ is the amount of task $\ell$ completed for \firm $i$ by the workers. Similarly, $y_j = (y_{j1}, \cdots, y_{jk}) \in \mathbb{R}_+^k$ denotes the allocation of the tasks to worker $j$, where $y_{j\ell}$ is the amount of task $\ell$ that is completed by worker $j$. Then, from this allocation \firm $i$ receives utility $u_i(x_i) = \sum_{\ell=1}^k v_{i\ell} x_{i\ell}$, and worker $j$ receives disutility or cost $d_j(y_j) = \suml c_{j\ell} y_{j\ell}$. 
Let $\mathbf{v} = (v_{i\ell})_{i\in [n], \ell \in[k]} \in \mathbb{R}_+^{n\times k},\mathbf{c} = (c_{j\ell})_{j\in[m], \ell \in[k]} \in \mathbb{R}_+^{m\times k}, \mathbf{x} = (x_{i\ell})_{i\in [n], \ell \in[k]} \in \mathbb{R}_+^{n\times k}$, and $\mathbf{y} = (y_{j\ell})_{j\in[n]}\in \mathbb{R}_+^{m\times k}$.

Both the \firmsS budgets and workers' earning requirements are exogenous to the model. The total budgets and the total earning requirements in the market are assumed to be equal, i.e., $\sumi B_i = \sumj E_j$. The \textit{prices} $\mathbf{p} =(p_1, \cdots, p_k) \in \mathbb{R}_+^k$ are to be assigned to tasks, where $p_\ell$ captures the per-unit price of task $\ell\in [k]$. 
Given the prices and allocations, we let $b_{i\ell} = p_\ell x_{i\ell}$ be the \textit{spending} of \firm $i$ to task $\ell$, and $e_{j\ell} = p_\ell x_{j\ell}$ be the \textit{earning} of worker $j$ from task $\ell$. Let $\mathbf{b} = (b_{i\ell})_{i\in [n], \ell \in[k]} \in \mathbb{R}_+^{n\times k},\mathbf{e} = (e_{j\ell})_{j\in[m], \ell \in[k]} \in \mathbb{R}_+^{m\times k}$. 

\begin{definition}[Labor Market Instance]
    The input to the labor market problem is a tuple $(n,m, k, \mathbf{v},\mathbf{c},\newline  (B_i)_{i=1}^n, (E_j)_{j=1}^m)$ where $n, m, k$ are numbers of users, workers, and tasks, $\mathbf{v} \in \mathbb{R}_+^{n\times k}, \mathbf{c}\in \mathbb{R}_+^{m\times k}$ are unit utilities and disutilities, $(B_i)_{i=1}^n, (E_j)_{j=1}^m$ are endowed budgets and earning requirements. The output is a tuple $(\mathbf{p}, \mathbf{x}, \mathbf{y})$, where $\mathbf{p}\in \mathbb{R}_+^n$ denote the prices of the tasks, and $\mathbf{x}, \mathbf{y} $ are allocations of the tasks to \firms and workers, respectively.
\end{definition}



The solution concept we consider is the \textit{competitive equilibrium} (CE), also called the market equilbirium, which is a tuple of prices and allocations under which all tasks requested by \firms are completed by the workers (with no redundant tasks that are completed but not requested), and the allocation to each agent belongs to their \textit{demand} (for users) or \textit{supply} (for workers), the set of optimal bundles under the prices. Specifically, the demand $\mathcal{D}_i^*(\p)$ of \firm $i$ is the set of bundles that maximize her utility subject to the budget constraint under $\p$, and the supply $\mathcal{S}_j^*(\p)$ of worker $j$ is the set of bundles that minimize her disutility while meeting the earning requirement under $\p$. 

\begin{definition}[Competitive Equilibrium]
    A triple $(\mathbf{p}, \mathbf{x}, \mathbf{y})$
    is a competitive equilibrium if and only if it satisfies
    \begin{itemize}
        \item \Firm optimality: 
        $
            \x_i \in \mathcal{D}_i^*(\p) = \underset{\x'_i \in \mathbb{R}^k_+}{\textnormal{argmax}} \Big\{ u_i(\x'_i) \;\; \textnormal{s.t.} \; \sum\nolimits_{\ell = 1}^k p_\ell x'_{i \ell} \leq B_i \Big\} \quad \forall\, i \in [n]; 
        $
        \item Worker optimality:
        $
            \y_j \in \mathcal{S}_j^*(\p) = \underset{\y'_j \in \mathbb{R}^k_+}{\textnormal{argmin}} \Big\{ d_j(\y'_j) \;\; \textnormal{s.t.} \; \sum\nolimits_{\ell = 1}^k p_\ell y'_{j \ell} \geq E_j \Big\}  \quad \forall\, j \in [m]; 
        $
        \item Market clearing: $
            \sum_{i = 1}^n x_{i\ell} = \sum_{j=1}^m y_{j\ell}\quad \forall\, \ell \in [k].
        $
    \end{itemize}
\end{definition}


\Firm optimality implies that each \firm requests those items that has the \textit{maximum bang per buck} (MBB) for her. Similarly, worker optimality requires that each worker gets items that has the \textit{minimum pain per buck} (MPB).
\begin{proposition}
    \label{prop: mbb/mpb}
    For any competitive equilibrium $(\mathbf{p}, \mathbf{x}, \mathbf{y})$, we have for any $i\in[n], j\in[m], k\in[n]$,
    \begin{equation*}
        x_{i\ell}>0 \implies \ell \in \arg \max_{\ell^\prime\in [k]} v_{i\ell^\prime}/p_{\ell^\prime}{\; \; \textnormal{and}\;\;}
        y_{j\ell}>0 \implies \ell \in \arg \min _{\ell^\prime \in[k]} c_{i\ell^\prime}/p_{\ell^\prime}.
    \end{equation*}
\end{proposition}

We also make some standard assumptions to avoid the trivial case where some prices are zero. First, we assume all unit disutilities $c_{j\ell}$ are strictly positive; otherwise worker $j$ could always provide infinite supply of task $\ell$. Second, we assume there are no trivial \firms or tasks, i.e., those giving receiving all-zero values. Under these assumptions, the equilibrium prices are always positive, and all budgets and earning requirements will be cleared exactly; we will assume that they hold throughout the paper.  

\begin{assumption}
\label{assumption: nontrivial}
    For each task $\ell \in [k]$, we have $v_{i\ell}>0$ for at least one \firm $i\in[n]$ and $c_{j\ell}>0$ for all $j\in[m]$. For each \firm $i\in[n]$, we have $v_{i\ell}>0$ for at least one task $\ell \in[k]$.
\end{assumption}

\begin{proposition}
    Under any input that satisfies \Cref{assumption: nontrivial}, the competitive equilibrium $(\mathbf{p}, \mathbf{x}, \mathbf{y})$ satisfies $\mathbf{p}>0$, and all budgets and earning requirements are cleared exactly, i.e., $\suml p_\ell x_{i\ell} = B_i$ for all $i\in[n]$ and $\suml p_\ell y_{j\ell} = E_j$ for all $j\in[m]$.
\end{proposition}
\begin{proof}
    For the first statement, suppose \Cref{assumption: nontrivial} holds but $p_\ell = 0$ for some $\ell \in[k]$. Since there exists $i\in[n]$ with $v_{i\ell}>0$, the demand of \firm $i$ for $\ell$ will be infinity, i.e., $x_{i\ell}^\prime = +\infty$ for any $x_i^\prime \in \mathcal{D}_i^*(\p)$. This leads to a contradiction: task $\ell$ should have zero supply $y_{j\ell} =0$ for all $j\in[m]$, since any worker suffers a positive disutility but earns nothing on $\ell$.
    
    For the second statement, notice that at any prices a \firm $i$ can always turn budgets into positive utility through some $\ell$ with $v_{j\ell}>0$, so allocations in $\mathcal{D}_i^*(\p)$ must deplete the budgets. Similarly, a worker $j$ will always suffer extra disutility whenever she increase her earning, so from the allocations in $\mathcal{S}_i^*(\p)$ the worker earns no more than the minimum earning requirement. 
\end{proof}

A natural immediate question is whether a competitive equilibrium always exist. In Appendix \ref{app:prop-existence} we show that it exist not only for linear utility and disutility functions, but much more general class of quasi-concave and quasi-convex utility and disutility functions (under mild non-satiation assumptions that are standard in the literature). Not only that, the sought-after optimality guarantees of the first and second welfare theorems are satisfied at CE (see Appendix \ref{app:prop-welfare-fairness}).

\paragraph{CEEEI and Fair Allocation.}
When all assigned budgets earning requirements are equal ($B_i =\cdots  = B_n =E_1=\cdots = E_m$), the equilibrium solution is known as \textit{competitive equilibrium with equal earning and incomes} (CEEEI). This special case is closely related to the notion of fair allocation, as it provides equal endowments on agents, potentially capturing the same level of purchasing or earning ability. In our labor market model, CEEEI indeed satisfies the fairness properties of envy-freeness and proportionality (see \Cref{thm: fair}).

%% file: sec-EG.tex
\section{Eisenberg-Gale Programs} 
\label{sec:EG-program}

{For the goods-only Fisher market, the celebrated Eisenberg-Gale (EG) convex program \cite{eisenberg1959consensus} facilitated a number of structural and computational results \cite{devanur2002market, cole2017convex}. \cite{bogomolnaia2017competitive} provided a similar EG-style non-convex formulation to capture CE for the chores-only Fisher market at the KKT points with strictly positive disutility for all the workers. In this section, we obtain a similar characterization of the CE set through an EG-style formulation, which, not so surprisingly, turns out to be non-convex. }


Consider the following primal Eisenberg-Gale type program. 
\begin{equation}
    \begin{aligned}
        \inf_{\bu, \bd > 0, \x, \y \geq 0} \quad & -\sum_{i = 1}^n B_i \log{u_i} + \sum_{j = 1}^m E_j \log{d_j} \\ 
        \textnormal{s.t.} \quad & u_i \leq u_i(\x_i) \quad \forall\, i \in [n] \\ 
        & d_j(\y_j) \leq d_j \quad \forall\, j \in [m] \\ 
        & \sum_{i=1}^n x_{i\ell} \leq \sum_{j=1}^m y_{j\ell} \quad \forall\, \ell \in [k].  
    \end{aligned}
    \tag{Primal EG}
    \label{prog:primal-EG-general}
\end{equation}

We show a one-to-one correspondence between the KKT points of~(\ref{prog:primal-EG-general}) and CE of the labor market for a broader class of utility functions. 

\begin{restatable}{theorem}{thmKKTPoints}
\label{thm: KKT-points}
    If the utilities and disutilities satisfy the following assumptions:
    \begin{enumerate}
        \item $u_i(x_i)$ is a continuously differentiable, concave, $1$-homogeneous function for each $i \in [n]$; 
        \item $d_j(y_j)$ is a continuously differentiable, convex, $1$-homogeneous function for each $j \in [m]$; 
        \item there exists at least one $i \in [n]$ such that $\frac{\partial u_i(x_i)}{\partial x_{ij}} > 0$ for any $x_i \geq 0$, for each $\ell \in [k]$, 
    \end{enumerate}
    then 
    every KKT point of~(\ref{prog:primal-EG-general}) corresponds to a competitive equilibrium of the labor market. 
\end{restatable}


Note that, even for the linear (dis)utility, it is difficult to find a stationary point of~(\ref{prog:primal-EG-general}), due to the 
existence of ``poles''. 
In particular, most of the classic optimization methods suffer from being attracted to some point (so-called pole) where $d_j(y_j) = 0$ for some worker $j$, which leads to the objective value of negative infinity but is not a stationary point. 
See discussion on similar issues in the chores market literature~\citep{boodaghians2022polynomial,chaudhury2024competitive}. 
Moreover, due to the non-convexity and unboundedness of~(\ref{prog:primal-EG-general}), it is difficult to extract meaningful information about CE. 
For example, even the existence of a CE cannot be guaranteed from this formulation. 

For the chores Fisher market with linear disutilities, a dual formulation of the primal EG program was recently shown to be very useful~\citep{chaudhury2024competitive}. 
Given the connection between our model and the chores Fisher market, 
we propose the following EG dual program for labor markets with linear (dis)utilities. 
\begin{equation}
    \begin{aligned}
        \max_{ \p, \mathbf{\bm \beta}, \mathbf{\bm \alpha} \geq 0} \quad & \sum_{i = 1}^n B_i \log{\beta_i} - \sum_{j = 1}^m E_j \log{\alpha_j} \\ 
        \textnormal{s.t.} \quad & v_{i\ell} \beta_i \leq p_\ell \quad \forall\, i \in [n], \ell \in [k] \\ 
        & p_\ell \leq d_{j l} \alpha_j \quad \forall\, j \in [m], \ell \in [k].  
    \end{aligned}
    \tag{Dual EG}
    \label{program:EG-Dual}
\end{equation}


The following lemma verifies that~(\ref{program:EG-Dual}) captures all CE of the linear good-chore Fisher market as its KKT points. 
\begin{restatable}{lemma}{lemDualKKT}
\label{lem: Dual-KKT}
    There is a one-to-one correspondence between the KKT points of~(\ref{program:EG-Dual}) and competitive equilibria of the linear labor market. 
\end{restatable}

Still, (\ref{program:EG-Dual}) falls into the regime of nonconvex optimization and thus is generally difficult to solve. 
Nevertheless, we observe that~(\ref{program:EG-Dual}) possesses many useful characteristics. 
For example, the feasible set of $(\mathbf{p}, \mathbf{\bm \beta}, \mathbf{\bm \alpha})$ is a \emph{polyhedral cone} excluding the origin (because of the implicit domain constraint). 
Also, the objective function is $0$-homogeneous: 
letting $f(\mathbf{p}, \mathbf{\bm \beta}, \mathbf{\bm \alpha}) = \sum_{i = 1}^n B_i \log{\beta_i} - \sum_{j = 1}^m E_j \log{\alpha_j}$, 
we have $f(k \mathbf{p}, k \mathbf{\bm \beta}, k \mathbf{\bm \alpha}) = f(\mathbf{p}, \mathbf{\bm \beta}, \mathbf{\bm \alpha})$ for any $k > 0$. 
This structure implies that if $(\mathbf{p}, \mathbf{\bm \beta}, \mathbf{\bm \alpha})$ is a KKT point of~(\ref{program:EG-Dual}), then so is $(c \mathbf{p}, c \mathbf{\bm \beta}, c \mathbf{\bm \alpha})$ for any scalar $c > 0$. 
Given this scaling invariance, we can, without loss of generality, 
fix the scale of $(\mathbf{p}, \mathbf{\bm \beta}, \mathbf{\bm \alpha})$ (e.g., $\sum_j p_j = 1$). 

%% file: sec-combinatorial-algorithm.tex
\section{Walrasian Algorithm for CE in Labor Market}
\label{sec:walrasian-algo}

In this section, we present a combinatorial, Walrasian algorithm that computes an exact CE for our labor market model in polynomial time. Without loss of generality, we assume all budgets $(B_i)_{i\in[n]}$ and earning requirements $(E_j)_{j\in[m]}$ are integral. We first introduce the graph structure the algorithm is based on.

\paragraph{The MBB/MPB Graph for Spending and Earning.} Our algorithm will work on a directed tripartite graph $\mathcal{G}(\p)$, namely \textit{MBB/MPB graph}, whose vertices correspond to the set of \firms on the left, tasks in the middle, and workers on the right. Edges are in the direction of \firms to tasks, or tasks to workers. Given price $\mathbf{p}$, an edge $\arc{i\ell}$ will present if and only if task $\ell$ is a MBB task for \firm $i$, and edge $\arc{\ell j}$ will present if and only if task $\ell$ is a MPB task for worker $j$. Denote the set of edges as $\mathcal{E}(\p)$.


The spending and earning in a competitive equilibrium can be modeled as a flow on the network. Formally, this will require setting the capacity of edges in $\mathcal{G}(\p)$ to $\infty$, and introducing a source/terminal node connected with each \firmBworker with edge capacity being the corresponding budget/earning requirement. Under this construction, the flow rate $\sumi B_i = \sumj E_j$. For general $\p$, any rate on the flow is upper-bounded by $\sumi B_i$, and $\p$ will be a equilibrium price if the maximum flow rate on $\mathcal{G}(\p)$ is $\sumi B_i$ (which implies market clearance together with flow conservation; optimality is ensured by the MBB/MPB edges).



\begin{claim}
\label{claim: spending/earning-to-allocation}
    Given prices $\p$, spending $(b_{i\ell})_{i\in[n], \ell \in[k]}$, and earning $(e_{j \ell })_{\ell \in[k], j \in[m]}$, then $(\p, \x, \y)$ gives a competitive equilibrium if
    \begin{itemize}
        \item $x_{i\ell } = b_{i\ell}/p_{\ell}, \ y_{j\ell } = e_{j \ell }/p_\ell$;
        \item \textnormal{\textbf{Flow conservation.}} Each \firm depletes her budget, each worker meets her earning requirement, and each task has zero surplus:
        \begin{equation*}
            \suml b_{i\ell} =  B_i \ \ \ \anyi,  \  \ \ \suml e_{j\ell} = E_j \ \ \ \anyj , \ \ \ \Delta_\ell  =0,\ \ \ \anyl.  
        \end{equation*}
        \item \textnormal{\textbf{Optimality.}} All non-negative money flow is supported only on MBB/MPB edges:
        \begin{equation*}
            e_{i\ell} >0 \implies \arc{i\ell} \in \mathcal{E}(\p), \  \ e_{j \ell } >0 \implies \arc{\ell j} \in \mathcal{E}(\p), \ \ \anyi, \anyj, \anyl. 
        \end{equation*}
    \end{itemize}
\end{claim}

\begin{algorithm}[t]
\caption{Algorithm for Computing CE in Labor Markets}
\label{algorithm: combinatorial}
\SetKwProg{Fn}{Function}{:}{}
\textbf{Initialization:} prices $\p =1^k$, threshold $\varepsilon >0$, lower bound parameter $\lambda>0$, $b_{i\ell} = B_i\cdot 1\{\ell =\min (\arg\max v_{i\ell} )\} , \anyi $,\ \  $e_{j\ell} = E_j\cdot 1\{\ell =\min (\arg\min c_{j\ell} )\}, \anyj$. 

\While{there exists $\ell \in [k]$ such that $|\Delta_{\ell}| > \varepsilon$}{

{\small \texttt{\color{blue} $\triangleright$ Definition of the large surplus set $S$}}

Set $S \gets \arg \max_{\ell\in [k]} \Delta_\ell$;
\While{$\max_{\ell \notin S} \Delta_{\ell} \geq (1-1/k) \min_{\ell \in S} \Delta_{\ell}$}{ $S \gets S \cup \{\ell\}$, where $\ell \in \argmax_{\ell \notin S} \Delta_{\ell}$;}

Set $\bar \Delta_{\ell} = \Delta_{\ell}$;

{\small \texttt{\color{blue} $\triangleright$ Spending and earning update}}

\While{there exists $i \in [n]$, $\ell \in S$, $\ell' \notin S$ such that $b_{i\ell} >0$ and $\arc{i\ell'} \in \mathcal{E}$}{
Decrease $b_{i\ell}$ and increase $b_{i \ell'}$ by the same amount until $b_{i\ell} = 0$ or $\Delta_{\ell} = \lb$:
\begin{equation*}
    \hat{b} = \min\left\{ b_{i\ell}, \Delta_\ell - \lb \right\}, \ \ b_{i\ell} \gets b_{i\ell} -\hat{b}, \ \ b_{i\ell'} \gets b_{i\ell'} + \hat{b}, \ \Delta_\ell \gets \Delta_\ell -\hat{b}, \ \Delta_{\ell^\prime} \gets \Delta_{\ell^\prime} + \hat{b}; 
\end{equation*}
\textbf{If} $\Delta_{\ell} = \lb$, \textbf{then} \textbf{Goto} Step 1\;}

\While{there exists $j \in [m]$, $\ell \in S$, $\ell' \notin S$ such that $e_{j\ell'} >0$ and $\arc{\ell j} \in \mathcal{E}$}{
    Decrease $e_{j\ell'}$ and increase $e_{j \ell}$ by the same amount until $e_{j\ell'} = 0$ or $\Delta_{\ell}^t = \lb$:
    \begin{equation*}
        \hat{e} = \min\left\{ e_{j \ell'}, \Delta_\ell - \lb\right\}, \ \ e_{j \ell'} \gets e_{j \ell'} -\hat{e}, \ \ e_{j \ell} \gets e_{j \ell} + \hat{e}, \ \Delta_\ell \gets \Delta_\ell -\hat{e}, \ \Delta_{\ell^\prime} \gets \Delta_{\ell^\prime} + \hat{e}; 
    \end{equation*}
    \textbf{If} $\Delta_{\ell} = \lb$, \textbf{then} \textbf{Goto} Step 1\;}

{\small \texttt{\color{blue} $\triangleright$ Price update}}

 Increase prices of all tasks in $S$ in proportion until a new MBB/MPB edge appears;

 Update the MBB/MPB edges $\mathcal{E}(\p)$;

\textbf{Goto} Spending and Earning Update (Step 7);

}
\Return{$\p$, $\mathbf{b}, \mathbf{e}$};
\end{algorithm}

Our Walrasian algorithm (\Cref{algorithm: combinatorial}) is based on this graph structure, which efficiently computes prices $\p$ and amounts of ``money flows'' $(b_{i\ell}), (e_{j\ell})$ such that 1) all positive spending and earnings are supported on MBB/MPB edges, 2) all budgets and earning requirements are cleared, i.e., $\sumi B_i = \sumj E_j = \sumi \suml b_{i\ell} = \sumj \suml e_{j\ell}$, 3) task surpluses are approximately zero, i.e., $|\Delta_\ell| < \varepsilon$ for all tasks $\ell$. We later show that an exact competitive equilibrium can be extracted when $\varepsilon$ is sufficiently small. 

The main idea of \Cref{algorithm: combinatorial} is to iteratively push the task surpluses to zero while preserving optimality and budget/earning requirement clearance; surpluses over all tasks will then sum up to zero always since the entire market's total budgets are equal to the total earning requirements. In each \textit{iteration} of the algorithm (a loop in Step 1), we pick a large surplus task subset $S \subseteq [k]$ by repeatedly adding the tasks to $S$ in the decreasing order of surplus until a multiplicative gap of $(1-1/k)$ appears; $S$ will be a nonempty proper subset of $[k]$ unless all surpluses are zero. Within an iteration, it alternates between two phases, \textit{spending and earning update} and \textit{price update}, until some task in $S$ has its surplus decreased ``sufficiently'', i.e., hits a lower bound of $\lambda \in (0,1)$ fraction of the surplus value at the beginning of the iteration (denoted as $\bar{\Delta}_\ell$ and recorded by the algorithm in Step 6). We also remark that \Cref{algorithm: combinatorial} can be initialized at any positive prices $\p$ and money flows on MBB/MPB edges; for analysis we assume a naive initialization $\p = 1^k$ and all \firms and workers spend/earn entirely from a single MBB/MPB task (lexicographical tie-breaking). 

\paragraph{Spending and Earning Update.} To decrease the surplus on some task $\ell \in S$, each \firm $i$ checks if there is an MBB edge $\arc{i\ell^\prime } \in \mathcal{E}(\p)$ to a task $\ell^\prime$ that is not is $S$. If so, then she can deviate her spending by decreasing $b_{i\ell}$ and increasing $b_{i\ell^\prime}$ by the same amount. Similarly, each worker $j$ checks if she can increase $e_{j\ell}$ and decrease $e_{j\ell^\prime}$ by the same amount for some $\arc{\ell^\prime j}\in \mathcal{E}(\p)$. We repeat these updates until the surplus $\Delta_\ell$ some $\ell \in S$ hits the lower bound $\lb$, and move on to the next iteration. However, this does not always happen: Since the updates are restricted to MBB/MPB edges only, it is possible that the lower bound cannot be reached even after all possible updates are performed on both sides. For this case, we proceed to the price update phase, which will provide more space for spending and earning updates. 

\paragraph{Price Update.} In the price update phase, we increase the price of all tasks in $S$ in proportion, i.e., scale them up with the same multiplier, until a new MBB/MPB edge appears in the graph. Intuitively, this update makes the subset $S$ less attractive to the \firms and more attractive to the workers, relative to the complement. This creates potentially more MBB/MPB edges for the spending and earning update phase to operate on. 

\paragraph{Correctness and Termination.} For correctness, we ensure that the budget/earning requirement clearance and optimality properties are indeed preserved throughout the two phases. For the spending and earning update phase, both properties are obvious since each unit operation increases and decreases money flows on MBB/MPB edges by the same value. For the price update phase, notice that $(b_{i\ell}), (e_{j\ell})$ are invariant so we only need to argue that all nonzero flows are still carried only by MBB/MPB edges. To see this, we show the following lemma for the price update phase.

\begin{restatable}{lemma}{lemPriceUpdate}
     \label{lem:price-update}
     During each price update, if an edge disappears from $\mathcal{E}(\p)$, then it carries zero flow. 
\end{restatable}

For termination, we first show that each iteration contains at most polynomially many phases of updates; each can be performed in polynomial time. Combined with the subsequent convergence rate analysis, the termination of the entire algorithm will be clear.

\begin{restatable}{lemma}{lemAllocationUpdate}
    \label{lem:allocation-update}
    An iteration has at most $n+m+1$ spending and earning update phases, and at most $n+m$ price update phases. Each of these phases terminates, and can be performed with polynomial (in $m$ and $n$; independent of $(B_i)_{i\in[n]}, (E_j)_{j\in[m]},\mathbf{v}$ and $\mathbf{c}$) number of arithmetic operations.
\end{restatable}



\subsection{Convergence Analysis} 

We next show how \Cref{algorithm: combinatorial} converges to zero surpluses in terms of $t$, the number of iterations. We measure the convergence by the sum of squared surpluses. Concretely, let $\bar\Delta_\ell^t$ denote the surplus of task $\ell$ at the beginning of the $t$'th iteration. Define the potential value of the iteration $t$ as $\phi^t = \suml (\bar\Delta_\ell^t)^2$. We will show that each iteration decreases the potential value multiplicatively:
\begin{lemma}
    \label{lem:convergence}
    When $\lambda = 1-0.09k^{-4}$, it holds for every iteration that $\phi^{t+1} \leq (1-0.01 k^{-7})\phi^t$.
\end{lemma}

\begin{proof}
The multiplicative gap in the definition of $S$ and the lower bound $\lb$ give us useful properties for the convergence analysis. 
    \begin{claim} 
    \label{claim:surplus_in_S}
    Let $S$ be the set constructed at each iteration in~\Cref{algorithm: combinatorial}. 
    Then, 
    \begin{itemize}
        \item For any two tasks $\ell, \ell' \in S$, we have $\bar \Delta^t_{\ell} \geq \bar{\Delta}^t_{\ell'}/e$; 
        \item For any two tasks $\ell \in S, \ell' \notin S$, we have $\bar \Delta^t_{\ell} > \bar \Delta^t_{\ell'} / (1 - \frac{1}{k})$. 
    \end{itemize} 
\end{claim}
To see the first statement, notice that $(1-1/k)^{k-1}$ is decreasing in $k$ and lower-bounded by $1/e$. Then for any $\ell, \ell^\prime \in S$, we have $\bar \Delta_\ell \geq \left( 1-1/k \right)^{|S|} \bar\Delta_{\ell'} \geq \left( 1- 1/k \right)^{k-1} \bar \Delta_{\ell'} \geq \bar \Delta_{\ell'}/e.$ Next, following from the definition and termination condition of the price update, we have
\begin{claim}
    \label{claim:bounds_on_delta_ell}
      Let $\delta_\ell^t = |\bar\Delta_\ell^t - \bar{\Delta}_\ell^{t+1}|$. Then
      \begin{itemize}
          \item For $\ell \in [k]\backslash S$ we have $\delta_\ell^t =\bar\Delta_\ell^{t+1} -\bar{\Delta}_\ell^t$. For $\ell \in S$ we have $\delta_\ell^t =\bar{\Delta}_\ell^t-\bar\Delta_\ell^{t+1}.$
          \item $\sum_{\ell \in S} \delta_\ell^t = \sum_{\ell \not \in S} \delta_\ell^t$, 
          \item For all $\ell \in S$ we have  $\delta_{\ell}^t \leq (1-\lambda)\bar \Delta^t_{\ell} $. There exists at least one $\ell \in S$, such that $\delta_{\ell}^t = (1-\lambda)\bar \Delta^t_{\ell} $.
      \end{itemize}
\end{claim}

Now we can bound the potential value after an iteration as follows. 
    \begin{align}
         \phi^{t+1} 
         =\sum_{\ell = 1}^k(\bar \Delta^{t+1}_{\ell})^2 
         \nonumber &=\sum_{\ell \in S}(\bar \Delta^t_{\ell} -\delta_{\ell}^t)^2 + \sum_{\ell \notin S}(\bar \Delta^t_{\ell} +\delta_{\ell}^t)^2 \\
         \nonumber&=\sum_{\ell = 1}^k (\bar \Delta^t_{\ell})^2 + \sum_{\ell = 1}^k (\delta_{\ell}^t)^2 -2\sum_{\ell \in S} \bar \Delta^t_{\ell} \delta_{\ell}^t + 2\sum_{\ell \notin S} \bar \Delta^t_{\ell} \delta_{\ell}^t \\
         \nonumber&\leq \sum_{\ell = 1}^k (\bar \Delta^t_{\ell})^2 + \sum_{\ell = 1}^k (\delta_{\ell}^t)^2 -2\min_{\ell \in S} \bar \Delta^t_{\ell}\sum_{\ell \in S}  \delta_{\ell}^t + 2\max_{\ell \notin S} \bar \Delta^t_{\ell} \sum_{\ell \notin S}  \delta_{\ell}^t \\
         \nonumber&\overset{\text{(a)}}{=}\sum_{\ell = 1}^k (\bar \Delta^t_{\ell})^2 + \sum_{\ell = 1}^k (\delta_{\ell}^t)^2 -2(\min_{\ell \in S} \bar \Delta^t_{\ell} - \max_{\ell \notin S} \bar \Delta^t_{\ell}) \sum_{\ell \in S} \delta_\ell^t \\
         &\overset{\text{(b)}}{\leq} \sum_{\ell = 1}^k (\bar \Delta^t_{\ell})^2 + \underbrace{\sum_{\ell = 1}^k (\delta_{\ell}^t)^2}_{\text{(I)}}- \underbrace{\frac{2}{k}(\min_{\ell \in S} \bar \Delta^t_{\ell} ) \sum_{\ell \in S}\delta_\ell^t}_{\text{(II)}},
         \label{eq: convergence-proof-1}
    \end{align}
    where step (a) is by the second statement of \Cref{claim:surplus_in_S} 
    and step (b) follows from 
    second statement of \Cref{claim:surplus_in_S}. Next, we bound the last two terms (I) and (II).
    
    \paragraph{Upper-bounding (I)} Since $\sum_{\ell \in S} \delta_\ell^t = \sum_{\ell \not \in  S} \delta_\ell^t$, by the equivalence between the $\ell_1$ and $\ell_2$ norms, 
    \begin{equation}
    \label{eq: convergence-proof-norm}
        \sum_{\ell =1}^k(\delta_\ell^t)^2 \leq \sum_{\ell  \in  S} (\delta_\ell^t)^2+  \sum_{\ell \not \in  S} (\delta_\ell^t)^2 \leq (k+1) \sum_{\ell  \in  S} (\delta_\ell^t)^2.
    \end{equation}
    Applying the third statement of \Cref{claim:bounds_on_delta_ell} to upper-bound $\delta_\ell^t$ for each $\ell\in S$, we get
    \begin{equation}
    \label{eq: convergence-proof-term-I}
        \text{(I)} = \sum_{\ell\in S} (\delta_{\ell}^t)^2 + \sum_{\ell\notin S} (\delta_{\ell}^t)^2  \leq (k+1) \cdot \sum_{\ell\in S} (\delta_{\ell}^t)^2 \leq (k+1)(1-\lambda)^2 \sum_{\ell\in S}(\Bar{\Delta}_\ell^t)^2 \leq (k+1)(1-\lambda)^2 \sum_{\ell=1}^k(\Bar{\Delta}_\ell^t)^2.
    \end{equation}
    
    \paragraph{Lower-bounding (II)} As in the third statement of \Cref{claim:bounds_on_delta_ell}, we know that there exists $\ell' \in S$ such that the surplus hits the lower bound. By the non-negativity of each $\delta_\ell^t$, we have $\sum_{\ell \in S} \delta_\ell^t \geq  \delta_{\ell'}^t = (1-\lambda) \Delta_{\ell'}^t$. We then have
    \begin{equation}
        \text{(II)} \geq \frac{2(1-\lambda)}{k}(\min_{\ell \in S} \bar{\Delta}_\ell^t) \cdot \bar{\Delta}_{\ell^\prime}^t \geq \frac{2(1-\lambda)}{k}(\min_{\ell \in S} \bar{\Delta}_{\ell}^t )^2 \overset{\text{(c)}}{\geq} \frac{2(1-\lambda)}{e^2k}(\max_{\ell \in [k]} \bar{\Delta}_{\ell}^t )^2 \overset{\text{(d)}}{\geq} \frac{2(1-\lambda)}{e^2 k^2(k-1)} \sum_{\ell = 1}^k (\Delta_\ell^t)^2.
        \label{eq: convergence-proof-term-II}
    \end{equation}
    The above step (c) is by the first statement in \Cref{claim:surplus_in_S}, and the fact that the maximum surplus value is always included in $S$. The last step (d) is by the following lemma; the proof is elementary. 
    
    \begin{lemma}
        For any $k \geq 2$ real numbers $y_1, \ldots, y_k$ such that $\sum_{\ell=1}^k y_\ell = 0$, it holds that $\sum_{\ell = 1}^k (y_\ell)^2 \leq k(k-1) (\max_{\ell\in[k]}y_\ell)^2$. 
    \end{lemma}
    
    \paragraph{Putting together.} Now we can bring~\cref{eq: convergence-proof-term-I} and~\cref{eq: convergence-proof-term-II} back to~\cref{eq: convergence-proof-1}. Recall by definition $\phi^t = \sum_{\ell=1}^k(\bar{\Delta}_\ell^t)^2$, then  
    \begin{equation*}
        \phi^{t+1} \leq \left(1+(k+1)(1-\lambda)^2- \frac{2(1-\lambda)}{e^2k^2(k-1)}\right) \phi^{t}.  
    \end{equation*}
    Setting $\lambda = 1-0.09 k^{-4}$ gives us the desired bound $\phi^{t+1} \leq (1-0.01 k^{-7}) \phi^t$. 
\end{proof}

\subsection{Extracting an Exact Equilibrium}

Finally, we show that when all surpluses are less than $1/k$ in absolute value, the output prices of \Cref{algorithm: combinatorial} will be equilibrium prices. 
\begin{lemma}
\label{lemma: combinatoria-extraction}
    When $\varepsilon<1/k$, the output prices of \Cref{algorithm: combinatorial} are competitive equilibrium prices.
\end{lemma}

\begin{proof}
    Let $\hat{\p}$ be the output prices, and $(\hat{b}_{i\ell}), (\hat{e}_{j\ell})$ be the output spending and earning from \Cref{algorithm: combinatorial}. We show that the max-flow of the MBB/MPB graph $\mathcal{G}(\hat{\p})$ is $\sumi B_i$, or equivalently $\sumj E_j$. Then $\hat{\p}$ and the max-flow will lead to a CE by \Cref{claim: spending/earning-to-allocation}.

    By total unimodularity we know that the max-flow rate of $\mathcal{G}(\hat{\p})$ is integral, so it suffices to show that there exists a flow with rate strictly larger than $\sumi B_i -1$. In fact, such flow can be constructed from the output by decreasing the \firmsS spending on $\Delta_\ell$ tasks and decreasing the workers' earning on $\Delta_\ell$ tasks until surpluses are zero. A concrete construction is  
    \begin{equation*}
        b_{i\ell}^* = \begin{cases}
            \hat{b}_{i\ell}\cdot  \left(1-\frac{\Delta_\ell}{\sum_{i^\prime} \hat{b}_{i^\prime\ell}}\right), & \text{($\Delta_\ell >0$)}\\
            \hat{b}_{i\ell}, & (o.w.)
        \end{cases}, \ 
        e_{j\ell}^* = \begin{cases}
            \hat{e}_{j\ell}\cdot  \left(1+\frac{\Delta_\ell}{\sum_{i^\prime} \hat{e}_{j^\prime\ell}}\right), & \text{($\Delta_\ell <0$)}\\
            \hat{e}_{j\ell}, & (o.w.)
        \end{cases}.
    \end{equation*}
    The flow rate is $\sumi B_i - \suml|\Delta_\ell| > \sumi B_i -1$ when $|\Delta_\ell|< \varepsilon <1/k$. 
\end{proof}

Combined with the convergence rate in iterations from \Cref{lem:convergence}, we know that an exact equilibrium price vector $\p$ can be computed by polynomial iterations of \Cref{algorithm: combinatorial}, each of them taking polynomial number of operations. An equilibrium allocation can then be computed efficiently by finding a max-flow on $\mathcal{G}(\p)$. 
\begin{theorem}
    When $\lambda = 1 - 0.09 k^{-4}$, after $t>100 k^{7}(\log \phi^1 + 4 \log k)$ iterations, the output prices of \Cref{algorithm: combinatorial} are equilibrium prices, where $\phi_1 \leq (\sumi B_i)^2+(\sumj E_j)^2.$ The total number of arithmetic operations is $\mathcal{O}(\operatorname{poly}(m, n , k, \log \max_i B_i, \log \max_j  E_j)).$
\end{theorem}

\paragraph{Strong Polynomial Result for CEEEI} 

Next, we consider the special case of computing CEEEI prices when all endowments ($B_i$ and $E_j$) are equal. By the conic structure of equilibrium prices \Cref{thm: convexity-of-prices}, we can set all $B_i = E_j = 1$ for all \firms and workers while preserving the set of equilibrium prices. This removes the $\log \max_i B_i, \log \max_j E_j$ dependence in the complexity. Further, as shown in \Cref{lem:allocation-update}, the number of arithmetic operations of the algorithm does not depends on the size of input $\mathbf{v}$ and $\mathbf{c}$. This gives a strongly polynomial time algorithm. 

\begin{theorem}
   When all endowments $B_i$'s and $E_j$'s are equal, by setting $\lambda = 1 - 0.09 k^{-4}$, \Cref{algorithm: combinatorial} computes a set of exact equilibrium prices in strong polynomial time, with $\mathcal{O}(\operatorname{poly}(m, n , k))$ arithmetic operations (independent of the bit-length of the input). 
\end{theorem}

%% file: sec-LP-algorithm.tex
\section{Linear Program Characterization} 
\label{sec:LP}

Several convex programming formulations are well-known for the goods-only Fisher market. Motivated by the positive computational results in Section \ref{sec:walrasian-algo}, we attempt to obtain a convex program for the Labor market; indeed in Appendix \ref{app:prop-convex} we show that the CE prices form a convex set. 
In this section we show that, in fact we can obtain a {\em linear programing (LP)} formulation for the Labor market; this is surprising because the lack of or the need for an LP formulation even for the goods-only Fisher market is well-known \citep{devanur2002market,adsul2010simplex,garg2013simplex}. Our LP is obtained through taking the dual of the EG-like program~(\ref{prog:primal-EG-general}) and a particular change of variables! Although this LP involves irrational coefficients, we design {\em truncation} and {\em extraction} procedures to obtain a polynomial-time algorithm to find a CE via this LP.  

\subsection{Linear Programs in the Logarithm Space}
\label{subsec:LP-formulation}

First, we show how the dual EG program (\textnormal{\ref{program:EG-Dual}}) can be converted into an LP after taking the logarithm on both sides of the constraints. 
When taking the logarithm, we ignore those constraints with $v_{i\ell} = 0$ because they are redundant with $p_\ell \geq 0$. 
This transformation leads to the following LP. 
\begin{equation}
    \begin{aligned}
        \max_{p, \beta, \alpha \geq 0} \quad & \sum_{i = 1}^n B_i \log{\beta_i} - \sum_{j = 1}^m E_j \log{\alpha_j} \\ 
        \textnormal{s.t.} \quad & \log{v_{i \ell}} + \log{\beta_i} \leq \log{p_\ell} \quad \forall\, i \in [n], \ell \in L(i) \\ 
        & \log{p_\ell} \leq \log{c_{j \ell}} + \log{\alpha_j} \quad \forall\, j \in [m], \ell \in [k], 
    \end{aligned}
\end{equation}
where $L(i) = \{ \ell \in [k] \mid v_{i \ell} > 0 \}$ for all $i \in [n]$. 
Recall that all $c_{j\ell}$ are strictly positive, therefore all $\log(c_{j\ell})$ terms are finite. 
Note that this LP does involve irrational numbers in its constraints, thus its exact solution cannot be directly found in polynomial time. 

We further consider the dual of the above LP. 
Interestingly, this leads to the following LP which is alike to the well-known Shmyrev program~\citep{shmyrev2009algorithm}. 
Similar to the Shmyrev program, 
the primal variables $\mathbf{b}$ and $\mathbf{e}$ can be interpreted as the spending and earning vectors, respectively. 
\begin{equation}
    \begin{aligned}
        \min_{b, e \geq 0} \quad & \sum_{i} \sum_{\ell \in L(i)} (-\log{v_{i\ell}}) b_{i\ell} + \sum_{j}\sum_{\ell} (\log{c_{j\ell}}) e_{j\ell} \\ 
        \textnormal{s.t.} \quad 
        & \sum_{\ell \in L(i)} b_{i\ell} = B_i \quad \forall\, i \in [n] \\ 
        & \sum_{\ell} e_{j\ell} = E_j \quad \forall\, j \in [m] \\ 
        & \sum_{i: \ell \in L(i)} b_{i\ell} = \sum_j e_{j\ell} \quad \forall\, \ell \in [k].  
    \end{aligned}
    \tag{Labor Market LP}
    \label{program:EG Dual log Dual}
\end{equation}

Note that the KKT conditions do not preserve under the logarithm operation, 
which requires us to 
separately prove that the optimal solution set of~(\ref{program:EG Dual log Dual}) indeed corresponds to the set of competitive equilibria. 
In particular, any optimal primal solution yields an optimal spending and earning allocation, while the dual solution provides an equilibrium price vector (after taking the exponential). 
The proof of the following theorem is in~\cref{app:subsec:proofs-LP}.

\begin{restatable}{theorem}{thmLPEquivalence}
\label{thm: LP-Equivalence}
    There is a one-to-one correspondence between the optimal solutions of \eqref{program:EG Dual log Dual} and the competitive equilibria of the linear labor market.
\end{restatable}

The above LP also connects the labor market to a network flow problem: finding a CE of the labor market is then equivalent to solving a minimum-cost flow problem with irrational cost coefficients; see detailed discussions in \Cref{appendix: discussion-LP}.

\subsection{Polynomial time algorithm}

In~\cref{subsec:LP-formulation}, we demonstrate that one can compute a CE in the labor market by solving an LP, or equivalently a minimum-cost flow problem, with irrational coefficients. 
However, in principle, all standard polynomial‐time algorithms presume rational (finite‐encoded) data. 
That said, a finite Turing machine cannot exactly encode with irrational input data. 
See, for example,~\citet{karloff2008linear}. 
As our program has irrational coefficients in its objective\footnote{One might consider truncating the coefficients in the objective function to a certain precision, such that the resulting problem give the same optimal solution as the original. 
However, in~\cref{app:subsec:truncated-LP}, we point out a fundamental challenge for this approach. 
Therefore, it is unclear whether there is a way to efficiently find a ``good'' truncated problem. }, 
it is impossible to find an exact CE by simply solving \textnormal{(\ref{program:EG Dual log Dual})}. 
Similarly, a combinatorial algorithm cannot be exactly implemented if the input data cannot even be finitely represented.

In this section, we propose another polynomial-time algorithm to find a CE of the labor market. 
This algorithm involves solving a minimum-cost flow problem with rounded coefficients, and a non-trivial extraction procedure. 
The main idea is as follows: 
after rounding the objective coefficients in \textnormal{(\ref{program:EG Dual log Dual})} up to certain decimal bits, 
one can leverage a polynomial-time algorithm to solve a minimum-cost flow problem and attain a CE $(\tilde{\mathbf{p}}, \tilde{\mathbf{b}}, \tilde{\mathbf{e}})$ for the ``rounded instance''. 
From this solution, 
we can find an equilibrium price $\mathbf{p}$ of the original instance and preserve $(\tilde{\mathbf{b}}, \tilde{\mathbf{e}})$ as a valid equilibrium spending and earning flow.
Therefore, $(\mathbf{p}, \tilde{\mathbf{b}}, \tilde{\mathbf{e}})$ constitutes a CE of the original instance. 
In the following, we explain our procedure in order and summarize it in the end.
All proofs are deferred to~\cref{app:subsec:proofs-LP}.

We denote  
$B = \max_i B_i, E = \max_j E_j, U = \max_{i,\ell} v_{i \ell}, D = \max_{j, \ell} c_{j \ell}$ 
and define 
\begin{equation*}
    L := k^2\log_2{(U D)} + \log_2{U} + \log_2{D} + 2\log_2{k} + 8 \geq k^2 + 8. 
\end{equation*}
The first step is rounding the coefficients of \textnormal{(\ref{program:EG Dual log Dual})} up to $L$ decimal bits. 
We show that, 
the ``rounded instance'' is close to the original one in terms of utility and disutility values (utility data). 

\begin{restatable}{claim}{claimMultApprox}
    \label{claim:mult-approx}
    Let $(\tilde{v}_{i\ell})_{i \in [n], \ell \in [k]}, (\tilde{c}_{j\ell})_{j \in [m], \ell \in [k]}$ denote the utility and disutility values in the rounded instance. 
    Then, we have 
    $v_{i \ell} \geq \tilde{v}_{i \ell} \geq (1 + 2^{1-L})^{-1} v_{i \ell}$ and $c_{j \ell} \geq \tilde{c}_{j \ell} \geq (1 + 2^{1-L})^{-1} c_{j \ell}$. 
\end{restatable}



By solving the rounded LP, we obtain an equilibrium price of the rounded market instance 
with rational utility data $(\tilde{v}_{i \ell})_{i \in [n], \ell \in [k]}$ and $(\tilde{c}_{j \ell})_{j \in [m], \ell \in [k]}$. 
The second step 
is to \emph{adjust} the output price and construct a \emph{connected} MBB/MPB graph defined in~\cref{sec:walrasian-algo}. 
We denote the adjusted price as $\tilde{\mathbf{p}}$, 
the resulting graph as $\tilde{\mathcal{G}}(\tilde{\mathbf{p}})$ and its set of edges as $\tilde{\mathcal{E}}(\tilde{\mathbf{p}})$. 
The following lemma verifies this step. 
\begin{restatable}{lemma}{lemAdjusted}
\label{lem: graph-of-adjusted-prices}
    The output price can be adjusted to a price $\tilde{\mathbf{p}}$ such that $(1)$ $\tilde{\mathcal{G}}(\tilde{\mathbf{p}})$ is connected, and $(2)$ $\sum_{\ell \in [k]} \tilde{p}_{\ell} = 1$, in at most $k$ price scaling phases. 
\end{restatable} 

Next, by leveraging the structure of $\tilde{\mathcal{G}}(\tilde{\mathbf{p}})$, we are able to efficiently compute a price vector $\mathbf{p}$ such that the set of edges of the MBB/MPB graph $\mathcal{G}(\p)$ constructed by $(v_{i \ell})_{i \in [n], \ell \in [k]}$, $(c_{j \ell})_{j \in [m], \ell \in [k]}$ (the utility data in the original instance) and $\mathbf{p}$ is a superset of that of $\tilde{\mathcal{G}}(\tilde{\mathbf{p}})$. 
\begin{restatable}{lemma}{lemFindAnEqPriceInPolytime}
\label{lem:find-an-eq-price-in-polytime}
    We can find a price vector $\mathbf{p}$ satisfying $\mathcal{E}(\mathbf{p}) \supseteq \tilde{\mathcal{E}}(\tilde{\mathbf{p}})$ in polynomial time.
\end{restatable} 
The above lemma verifies that, 
given $\p$,
both $\tilde{\mathbf{b}}$ and $\tilde{\mathbf{e}}$ induce money flows only along optimal edges.
Hence, $(\mathbf{p}, \tilde{\mathbf{b}}, \tilde{\mathbf{e}})$ constitutes a CE of the original instance.
We summarize our truncation-and-extraction procedure as follows: 
\begin{enumerate}
    \item Encode~(\ref{program:EG Dual log Dual}) up to $L = k^2\log_2{(U D)} + \log_2{U} + \log_2{D} + 2\log_2{k} + 8$ decimal bits; 
    \item Solve the minimum-cost flow problem corresponding to~(\ref{program:EG Dual log Dual}); 
    \item Adjust the output prices by multiplicatively increase a subset of the prices in each component of the MBB/MPB graph, until the graph is connected; 
    \item Normalize the prices such that $\sum_{\ell \in [k]} p_\ell = 1$; 
    \item Use the original utility data as the input, solve a polynomial-size system of linear equations
    associated with the connected MBB/MPB graph (after the price adjustment); 
    \item The obtained solution is an equilibrium price, coupled with the optimal spending and earning allocation output by~(\ref{program:EG Dual log Dual}). 
\end{enumerate}



\begin{theorem}
    We can compute an exact CE of the labor market via solving a rounded LP and a truncation-and-extraction procedure. 
    The procedure can be done in polynomial time. 
\end{theorem}

%% file: appendix.tex
\appendix

\section{Properties of Equilibria} 

We prove the first and second welfare theorems, fairness properties on equal budgets/earnings, as well as the existence of CE for larger classes of utility and disutility functions beyond linear case. 
Every linear labor market satisfying~\cref{assumption: nontrivial} belongs to this class of labor markets. 
Moreover, for linear labor markets, we show that the set of equilibrium prices form a convex cone (without zero), and also provides a set of primal and dual Eisenberg-Gale programs to capture the competitive equilibria by KKT conditions. 


\subsection{Welfare Theorems and Fairness}
\label{app:prop-welfare-fairness}
\begin{definition}
    A function $f: \mathbb{R}^k \to \mathbb{R}$ is \emph{strongly increasing} if 
    $f(z') > f(z)$ whenever $z' \geq z$ and $z'_\ell > z_\ell$ for some $\ell$. 
\end{definition}
\begin{definition}
    A function $f: \mathbb{R}^k \to \mathbb{R}$ is \emph{locally non-satiated} if given any $z$, there exists a coordinate $\ell$ satisfying 
    $f(z') > f(z)$ if $z' \geq z$ and $z'_\ell > z_\ell$. 
\end{definition}
Obviously, any strongly increasing function is locally non-satiated. 
Intuitively, a \firm will use up her budget when they maximize their utility under budget constraint, if their utility function is locally non-satiated. 
Symmetrically, a worker will earn exactly her earning expectation when she minimizes her disutility with earning requirement, if her disutility function is locally non-satiated. 
We will repeatedly invoke this fact in proving the subsequent theorems. 

\begin{definition}
    A function $f: S \rightarrow \mathbb{R}$ is \emph{quasiconcave} if $f(\lambda x + (1 - \lambda)y) \geq \min\{ f(x), f(y) \}$ for any $x, y \in S$ and $\lambda \in [0, 1]$ and \emph{strictly quasiconcave} if $f(\lambda x + (1 - \lambda)y) > \min\{ f(x), f(y) \}$ for any $x, y \in S$ and $\lambda \in (0, 1)$. 
    Symmetrically, 
    a function $f: S \rightarrow \mathbb{R}$ is \emph{quasiconvex} if $f(\lambda x + (1 - \lambda)y) \leq \max\{ f(x), f(y) \}$ for any $x, y \in S$ and $\lambda \in [0, 1]$ and \emph{strictly quasiconvex} if $f(\lambda x + (1 - \lambda)y) < \max\{ f(x), f(y) \}$ for any $x, y \in S$ and $\lambda \in (0, 1)$. 
\end{definition}
Any concave (resp. convex) function is strictly quasiconcave (resp. strictly quasiconvex). 

\begin{assumption}
    We assume that: 
    \begin{itemize}
        \item All utility functions are continuous, locally non-satiated, and strictly quasiconcave; 
        \item All disutility functions are continuous, strongly increasing, and quasiconvex. 
    \end{itemize}
    \label{assumption:general-utility}
\end{assumption}

We call an allocation $(\x, \y)$ a \emph{feasible} allocation if the labor supply can cover all of the labor demand, i.e., $\sum_i x_{i\ell} \leq \sum_j y_{j\ell} <\infty$ for all $\ell\in[k]$. 
That is, all tasks demanded by the \firms must be completed by some of the workers; however, there may be excess labor supply for certain tasks---some workers may receive payment without actually performing part of their allocated work. (This is the case for Amazon Flex.)

\begin{definition}
    A feasible allocation \( (\x, \y) \) is said to be a \emph{Pareto optimal allocation} if there is no other feasible allocation \( (\x', \y') \) such that:
    \begin{enumerate}[label=\roman*.]
        \item \( u_i(x'_i) \geq u_i(x_i) \) for all \( i \in [n] \) and \( d_j(y'_j) \leq d_j(y_j) \) for all \( j \in [m] \), and 
        \item at least one of the above inequalities is strict, i.e., \( u_i(x'_i) > u_i(x_i) \) for some \( i \in [n] \) or \( d_j(y'_j) < d_j(y_j) \) for some \( j \in [m] \).
    \end{enumerate} 
\end{definition}

Note that there is a trivial Pareto optimal (PO) allocation $(0, 0)$---any nonzero allocation results in strictly greater disutility for the workers. 

Next, we prove the first welfare theorem for CE in labor markets. 
\begin{theorem}[First welfare theorem]
    If the labor market satisfies~\cref{assumption:general-utility}, then 
    for any competitive equilibrium $(\p, \x, \y)$, $(\x, \y)$ is a nonzero Pareto optimal allocation. 
\end{theorem}
\begin{proof}
    We prove this by contradiction. 
    Suppose otherwise, i.e., there is a feasible allocation $(\x', \y')$ such that $(i)$ $u_i(x'_i) \geq u_i(x_i)$ for all $i \in [n]$, $(ii)$ $d_j(y'_j) \leq d_j(y_j)$ for all $j \in [m]$, and $(iii)$ at least one unit improves their utility/disutility over $(\x, \y)$. 
    First, note that $(i)$ imples that $\inp{\p}{x'_i} \geq B_i$ for all $i \in [n]$. 
    This can be proven as follows: Suppose otherwise, i.e., $\inp{\p}{x'_i} < B_i$ for some \firm $i$. Then, given $\p$, the \firm $i$ can strictly improve their utility by saturating their budget, as the utility function is locally non-satiated by~\cref{assumption:general-utility}. 
    This contradicts $x_i \in \mathcal{D}^*_i(\p)$. 
    Symmetrically, $(ii)$ implies that $\inp{\p}{y'_j} \leq E_j$ for all $j \in [m]$. 
    
    Then, we consider two cases. 
    First, if there is a \firm $\hat{i}$ improving their utility, i.e., $u_{\hat{i}}(x'_{\hat{i}}) > u_{\hat{i}}(x_{\hat{i}})$, then we have $\inp{\p}{x'_{\hat{i}}} > B_i$ because $x^*_{\hat{i}}$ maximizes $u_{\hat{i}}(\cdot)$ over $\{z \in \mathbb{R}^k_+ \mid \inp{p^*}{z} \leq B_{\hat{i}} \}$. This yields that 
    \begin{equation}
        \sum\nolimits_i \inp{\p}{x'_i} = \sum\nolimits_{i \neq \hat{i}} \inp{\p}{x'_i} + \inp{\p}{x'_{\hat{i}}} > \sum\nolimits_i B_i = \sum\nolimits_j E_j \geq \sum\nolimits_j \inp{\p}{y'_j}. 
        \label{eq:first-welfare-theorem-proof-contradiction}
    \end{equation}
    Equivalently, $\sum\nolimits_j \inp{\p}{y'_j} - \sum\nolimits_i \inp{\p}{x'_i} < 0$, which contradicts $\p \geq 0$ and $\sum_j y'_{j\ell} - \sum_i x'_{i\ell} \geq 0$ (recall that $(\x', \y')$ a feasible allocation). 
    Second, if there is a worker $\hat{j}$ improving their disutility, i.e., $d_{\hat{j}}(y'_{\hat{j}}) < d_{\hat{j}}(y_{\hat{j}})$, then we have $\inp{\p}{y'_{\hat{j}}} < E_j$ because $y_{\hat{i}}$ minimizes $d_{\hat{j}}(\cdot)$ over $\{y \in \mathbb{R}^k_+ \mid \inp{\p}{y} \geq E_{\hat{j}} \}$. This leads to a contradiction, similar to~\cref{eq:first-welfare-theorem-proof-contradiction}. 
    Therefore, $(\x, \y)$ is a Pareto optimal allocation. 
    The allocation 
    $(\x, \y)$ is clearly nonzero, because $\sum_j y^*_j = 0$ implies $\sum_j \inp{\p}{y^*_j} = 0 < \sum_j E_j$ which contradicts the worker optimality. 
\end{proof}
The first welfare theorem provides a powerful statement, as it shows that the outcome of the competitive economy yields an efficient allocation from a social perspective. 
This indeed confirms that CE is a desirable outcome in the labor market.

The second welfare theorem states that 
any Pareto optimal allocation is supported a competitive equilibrium, 
once we correctly set money endowments. 
Clearly, the zero PO allocation $(\x, \y) = (0, 0)$ can be supported by a CE of the labor market with $\sum_i B_i = \sum_j E_j = 0$. 
We next prove the non-trivial case, by showing that any nonzero Pareto-optimal allocation can be supported by a CE of the labor market, up to a redistribution of money endowments. 

\begin{theorem}[Second welfare theorem]
    If the labor market satisfies~\cref{assumption:general-utility}, 
    then for every nonzero Pareto optimal allocation $(\hat{\mathbf{x}}, \hat{\mathbf{y}})$, there exists a distribution of budgets and incomes $(\hat{\mathbf{B}}, \hat{\mathbf{E}})$ such that \textnormal{(1)} $\sum_i \hat{B}_i = \sum_j \hat{E}_j = \sum_i B_i$ and \textnormal{(2)} 
    there exists a price vector $\hat{\mathbf{p}} \in \mathbb{R}^k_+$ such that $(\hat{\mathbf{p}}, \hat{\mathbf{x}}, \hat{\mathbf{y}})$ constitutes a CE. 
\end{theorem}
\begin{proof} 
    We prove this lemma by using the Hyperplane Separation Theorem. 
    Given any PO allocation $(\hat{\mathbf{x}}, \hat{\mathbf{y}})$, we first construct the \firmS (strictly) better-than allocation set: 
    \begin{equation*}
        X_i = \{ x_i \geq 0: u_i(x_i) \ge u_i(\hat{x}_i), \forall i\in[n];\ u_f(x_f) > u_f(\hat{x}_f), \exists f\in[n]  \} \subset \mathbb{R}^k_+ \text{ and } X = \sum\nolimits_i X_i, 
    \end{equation*}
    where the summation (here and throughout) represents Minkowski addition. 
    Because all utility functions are quasiconcave, 
    $X$ is an nonempty open convex set.  
    Similarly, we define 
    \begin{equation*}
        \bar{Y}_j = \{ y_j \geq 0: d_j(y_j) \leq d_j(\hat{y}_j) \} \subset \mathbb{R}^k_+ \text{ and } \bar{Y} = \sum\nolimits_j \bar{Y}_j. 
    \end{equation*}
    Here, $\bar{Y}$ is a nonempty compact convex set, because all disutility functions are quasiconvex and strongly increasing. 
    Since $(\hat{\mathbf{x}}, \hat{\mathbf{y}})$ is a PO allocation, $X \cap \bar{Y} = \emptyset$. 
    By the Hyperplane Separation Theorem, there exists a $\bar{\mathbf{p}} \in \mathbb{R}^k \setminus \{ 0 \}$ and an $\bar{c} \in \mathbb{R}$ such that 
    \begin{equation}
        \inp*{\bar{\mathbf{p}}}{x} \geq \bar{c} \hspace{10pt} \mbox{ for all } x \in X 
        \hspace{30pt} \textnormal{ and } \hspace{30pt}
        \inp*{\bar{\mathbf{p}}}{y} \leq \bar{c} \hspace{10pt} \mbox{ for all } y \in \bar{Y}. 
        \label{separating-hyperplane-inequalities}
    \end{equation}
   
    Recall that $x, y \in \mathbb{R}^k_+$ and $\hat{\mathbf{x}} \in \mathbb{R}^{nk}_+, \hat{\mathbf{y}} \in \mathbb{R}^{mk}_+$. 
    Since $\mathbf{0}_k \in Y$, $\bar{c} \geq 0$. 
    For each $\ell \in [k]$, since $\hat{x} + t \mathbf{e}_\ell \in X$ for all $t > 0$, where $\mathbf{e}_\ell$ is the unit vector with $1$ in the $\ell$th coordinate and $0$ otherwise,
    we have $\bar{p}_\ell \geq 0$ because otherwise $\inp*{\bar{\mathbf{p}}}{\hat{x} + t \mathbf{e}_\ell} < 0$ as $t \rightarrow \infty$. 
    Moreover, letting $\bar{\ell} \in \textnormal{argmax}_\ell\  \bar{p}_\ell > 0$, 
    there is a $y' = \varepsilon \mathbf{e}_{\bar{\ell}} \in \bar{Y}$ for some $\varepsilon > 0$ (we can only distribute this task among those workers with $d_j(\hat{y}) > d_j(\mathbf{0}_k)$), which ensures $\bar{c} > 0$. 
    Also, as scaling $(\bar{\mathbf{p}}, \bar{c})$ does not change the separating hyperplane, we can scale $(\bar{\mathbf{p}},\bar{c})$ to  $(\hat{\mathbf{p}},\hat{c})$ such that $\hat{c} = \sum_i B_i = \sum_j E_j$ and $\hat{\mathbf{p}} \in \mathbb{R}^k_+ \setminus \{ 0 \}$. 

    Because $(\hat{\mathbf{x}}, \hat{\mathbf{y}})$ is a feasible allocation, we have $\sum_i \hat{x}_{i\ell} \leq \sum_j \hat{y}_{j\ell}$ for all $\ell \in [k]$. 
    By~\cref{separating-hyperplane-inequalities}, it holds that 
    \begin{equation*}
        \sum\nolimits_i B_i \leq \inp{\hat{\mathbf{p}}}{\sum\nolimits_i \hat{x}_i} \leq \inp{\hat{\mathbf{p}}}{\sum\nolimits_j \hat{y}_j} \leq \sum\nolimits_j E_j = \sum\nolimits_i B_i. 
    \end{equation*}
    Therefore, $\sum_i \inp{\hat{\mathbf{p}}}{\hat{x}_i} = \sum_i B_i$ and $\sum_j \inp{\hat{\mathbf{p}}}{\hat{y}_j} = \sum_j E_j$. 
    Consequently, there is a redistribution of money endowments $\hat{B}_i := \inp{\hat{\mathbf{p}}}{\hat{x}_i}\;\forall\,i \in [n]$ and $\hat{E}_j := \inp{\hat{\mathbf{p}}}{\hat{y}_j}\;\forall\,j \in [m]$ such that $\sum_i \hat{B}_i = \sum_i B_i$ and $\sum_j \hat{E}_j = \sum_j E_j$. 

    Next, we show that $(\hat{\mathbf{p}}, \hat{\mathbf{x}}, \hat{\mathbf{y}})$ is a CE with the money endowments $(\hat{\mathbf{B}}, \hat{\mathbf{E}})$. 
    To prove the \firm and worker optimality, we need the following fact. 

    To the contrary, suppose $(\hat{\mathbf{x}}, \hat{\mathbf{y}})$ does not form an optimal bundle of \firms and workers at prices $\hat{p}$ with respect to endowments $(\hat{\mathbf{B}},\hat{\mathbf{E}})$. Instead, let $(\mathbf{x}',\mathbf{y}')$ be an optimal bundle. We have $u_i(x'_i) \ge u_i(\hat{x}_i), \forall i\in [n]$, as each \firm $i$ can afford $\hat{x}_i$ at prices $\hat{\mathbf{p}}$. Similarly, $d_j(y'_j) \le d_j(\hat{y}_j), \forall j \in [m]$. 
    Since, $(\hat{\mathbf{x}},\hat{\mathbf{y}})$ isn't an optimal bundle one of the \firm or worker is strictly better off at $(\mathbf{x}', \mathbf{y}')$. Wlog let this be \firm $i$, that is, $u_i(x'_i)>u_i(\hat{x}_i)$. 
    
    Since $x'_i$ is an optimal bundle of \firm $i$ at prices $\hat{\mathbf{p}}$, either $\langle \hat{\mathbf{p}}, x'_i\rangle < \hat{B}_i$, or there must a task $\ell\in[k]$ such that $p_\ell>0$ and $i$ is non-satiated for $\ell$ at $x'$. In the former case, construct a new allocation $\tilde{\mathbf{x}}$ such that $\tilde{x}_i=x''_i$, and $\tilde{x}_f=\hat{x}_f$ for all $f\neq i$. In the latter case, $i$ can consume $\ell$ by a tiny bit less while keeping utility higher than $u_i(\hat{x}_i)$. That is, $\exists x''_i$ such that $u_i(x'_i) > u_i(x''_i) > u_i(\hat{x}_i)$ and $\langle \hat{\mathbf{p}},x''_i\rangle < \hat{B}_i$. Now construct allocation $\tilde{x}$ where $\tilde{x}_i=x''_i$, and $\tilde{x}_f=\hat{x}_f$ for all $f\neq i$. 
    
    In either case, $\tilde{\mathbf{x}}\in X$ while $\langle \hat{\mathbf{p}},\tilde{\mathbf{x}}\rangle =\sum_f \langle \hat{\mathbf{p}},\tilde{x}_f\rangle < \hat{B}_i +\sum_{f\neq i} \hat{B}_f = \hat{c}$, contradicting \eqref{separating-hyperplane-inequalities}.
\end{proof}

\paragraph{Fairness properties of CEEEI.}
We show that CEEEIs in our market model satisfies the fairness properties of envy-freeness and proportionality. \textit{Envy-freeness} requires that no agent prefers someone else's allocation to her own, and is implied by the fact that each \firm and worker gets an optimal bundle. Under concavity of utility functions (resp. convexity of disutility functions), this further implies \textit{proportionality}, which requires that each \firm gets at least $1/n$ of her total utility of the gross bundle of all available tasks, and each worker suffers at most $1/m$ of his disutility utility of the gross bundle of all requested tasks. 

\begin{theorem}
    \label{thm: fair}
    Suppose $B_1 = \cdots = B_n$ and $E_1 = \cdots= E_m$. Then for any utility function, any competitive equilibrium allocation $(\mathbf{x}, \mathbf{y})$ satisfies \em{envy-freeness}, i.e.,
    \begin{equation*}
        u_i(x_i) =\arg \max_{i^\prime\in [n] } u_i(x_{i^\prime})\;\text{ for all $i\in[n]$} \;\; \text{and} \;\; d_j(y_j) = \arg \min _{j^\prime \in [m]} d_j(y_{j^\prime})\;\text{ for all $j\in[m]$.}
    \end{equation*}
    Further, when each $u_i(\cdot)$ is concave and each $d_j(\cdot)$ is convex, and $u_i(0) = d_j(0) = 0$, any competitive equilibrium allocation $(\mathbf{x}, \mathbf{y})$ satisfies \em{proportionality}, i.e., 
    \begin{equation*}
        u_i(x_i) \geq \frac{1}{n} u_i\left(\sumj y_j\right) \;\text{ for all $i\in[n]$} \;\; \textnormal{and} \;\;d_j(y_j) \leq \frac{1}{m} d_j\left(\sumi x_i\right) \;\text{ for all $j\in[m]$}.
    \end{equation*}
\end{theorem}
\begin{proof}
    With equal budgets, envy-freeness follows directly from the fact that $x_i\in \mathcal{D}_i^*(\p)$ for each \firm $i$ and $y_j \in \mathcal{S}_j^*(\p)$ for each worker $j$. For proportionality of the \firms, we have
    \begin{equation*}
        u_i(x_i) \geq  u_i\left(\frac{1}{n} \sum_{i^\prime=1}^n x_{i^\prime}\right) \geq \frac{1}{n} u_i\left(\sum_{i^\prime=1}^n x_{i^\prime}\right) =\frac{1}{n} u_i\left(\sum_{j^\prime=1}^m y_{j^\prime}\right),
    \end{equation*}
    where the first inequality is by \firm $i$'s individual optimality (notice that the bundle $(1/n) \sum_{i^\prime=1}^n x_{i^\prime}$ will cost \firm $i$ exactly $B_i$ under equal incomes), the second inequality is by concavity, and the final equality is by market clearance. The proportionality of workers can be shown similarly. 
\end{proof}



\subsection{Existence of Competitive Equilibrium}\label{app:prop-existence}

We prove the existence of CE for a fairly large class of utility and disutility functions. 
The proof is based on Kakutani's fixed point theorem~\citep{kakutani1941generalization}. 
To prove the existence, we need an additional nonsatiation-style assumptions. 
Intuitively, we need to ensure that given any allocation $x$, 
for each task $\ell$, there is a \firm whose utility increases once allocated more task $\ell$. 
\begin{assumption}
    For any $\ell \in [k]$, there is at least one $i \in [n]$ such that $u_i(x'_i) > u_i(x_i)$ whenever $x'_i \geq x_i$ and $x'_{i\ell} > x_{i\ell}$. 
    \label{assumption:at-least-one-firm-demand}
\end{assumption}

We remark that every linear labor market satisfying~\cref{assumption: nontrivial} satisfies~\cref{assumption:general-utility,assumption:at-least-one-firm-demand}.

\begin{theorem}[Existence of CE]
    \label{thm: existence}
    If the labor market satisfies~\cref{assumption:general-utility,assumption:at-least-one-firm-demand}, then there exists a competitive equilibrium of the labor market. 
\end{theorem}
\begin{proof}
    First, we construct a \emph{nonempty compact and convex} (NCC) set. 
    To this end, we consider the price vectors in the unit simplex, i.e., $\mathbf{p} \in P := \{ \mathbf{p} \in \mathbb{R}^k_+ \mid \sum_{\ell=1}^k p_\ell = 1 \}$. 
    For any $\mathbf{p} \in P$, $\mathcal{S}^*_j(\mathbf{p})$ is NCC. 
    This is because we can always pick $\hat{\ell} \in \textnormal{argmax}_{\ell}  p_\ell$, (note that $p_{\hat{\ell}} \geq \frac{1}{k}$), and then let $\hat{y}_j = k E_j \mathbf{e}_{\hat{\ell}} \geq \frac{E_j}{p_{\hat{\ell}}} \mathbf{e}_{\hat{\ell}}$. 
    Since $\hat{y}_j$ is feasible to the disutility minimization problem, 
    $\mathcal{S}^*_j(\mathbf{p}) \subset \{ y_j \mid d_j(y_j) \leq d_j(k E_j \mathbf{e}_{\hat{\ell}}) \}$. 
    We denotes $\bar{S} := \max_{j \in [m]} \max\{ \norm*{y_j}_\infty \mid d_j(y_j) \leq d_j(k E_j \mathbf{e}_{\hat{\ell}}) \} < \infty$. 
    However, it is possible that the demand correspondence is unbounded if some price goes to zero. 
    To handle this, we define ``truncated'' demand correspondences 
    \begin{equation*}
        \bar{\mathcal{D}}_i^*(\mathbf{p}) := \textnormal{argmax}_{x'_i \in \mathbb{R}^k_+} \Big\{ u_i(x'_i) \;\; \textnormal{s.t.} \; \sum\nolimits_{\ell = 1}^k p_\ell x'_{i \ell} \leq B_i \text{ and } \norm*{x_i'}_\infty \leq \bar{D} \Big\} \quad \forall\, i \in [n], 
    \end{equation*}
    where $\bar{D} = \max\{ 2\bar{S}, 2\sum_i B_i \}$. 
    Because all utility functions are continuous and quasiconcave by~\cref{assumption:general-utility}, given any $\mathbf{p} \in P$, each $\bar{\mathcal{D}}_i^*(\mathbf{p})$ is NCC. 
    In particular, any such correspondence maps $\mathbf{p} \in P$ to $\mathcal{B}_\infty(\bar{D})$, where $\mathcal{B}_\infty(r)$ denotes the $\ell_\infty$ ball centered at $0$ with radius $r$. 
    Then, we define the ``truncated'' excess demand correspondence 
    \begin{equation*}
        \bar{\mathcal{Z}}^*(\mathbf{p}) := \left\{ z \in \mathbb{R}^k \left\vert\, z_\ell = \sum\nolimits_i x^*_{i\ell} - \sum\nolimits_j y^*_{j\ell} \; \forall\, \ell \textnormal{ where } x^*_i \in \bar{\mathcal{D}}_i^*(\mathbf{p}) \; \forall\, i \textnormal{ and } y^*_j \in \mathcal{S}_j^*(\mathbf{p}) \; \forall\, j \right. \right\}. 
    \end{equation*}
    As the set $\bar{\mathcal{Z}}^*(\mathbf{p})$ is the Minkowski summation of $n + m$ NCC sets, it is therefore NCC as well. 
    In particular, $\bar{\mathcal{Z}}^*$ maps $\mathbf{p} \in P$ to $Z := \mathcal{B}_\infty(\max\{n\bar{D}, m\bar{S}\})$. 
    
    Next, we define a best-response price correspondence $\Phi: Z \rightrightarrows P$: 
    \begin{equation*}
        \Phi(z) = {\textnormal{argmax}}_{\mathbf{p}' \in P} \inp*{\mathbf{p}'}{z}. 
    \end{equation*} 
    It is clear that $\Phi(z)$ is a convex set for any $z \in Z$. 
    Also, note that all correspondences $\bar{\mathcal{D}}_i^*, \mathcal{S}^*_j$ are the set of maximizers of a continuous function over a feasible set $C(\mathbf{p})$ where $C$ is a continuous, compact-valued correspondence. 
    By Berge's Maximum theorem~\citet{Ber63}, all correspondences $\bar{\mathcal{D}}_i^*, \mathcal{S}^*_j$ are upper hemi-continuous. 
    Similarly, $\Phi$ is also upper hemi-continuous. 
    We now arrive at our target correspondence $\Psi: P \times Z \rightrightarrows P \times Z$ defined as 
    \begin{equation*}
        \Psi(\mathbf{p}, z) = \left( \Phi(z), \bar{\mathcal{Z}}^*(\mathbf{p}) \right). 
    \end{equation*}
    Since (1) $P \times Z$ is NCC; 
    (2) $\Psi$ is upper hemi-continuous everywhere on $P \times Z$; 
    (3) $\Psi(\mathbf{p}, z)$ is a convex set for all $(\mathbf{p}, z) \in P \times Z$. 
    By Kakutani's fixed point theorem, there exists a pair of $(p^*, z^*)$ satisfying $(\mathbf{p}^*, z^*) \in \Psi(\mathbf{p}^*, z^*)$, i.e., $\mathbf{p}^* \in \Phi(z^*)$ and $z^* \in \bar{\mathcal{Z}}^*(\mathbf{p}^*)$. 

    Next, we show that $(\mathbf{p}^*, z^*)$ corresponds to a CE of the ``untruncated'' market. 
    Consider a fixed point $(\mathbf{p}^*, z^*)$. 
    We first show that $p^*_\ell > 0 \;\forall\, \ell \in [k]$. 
    Suppose otherwise, there is a set of tasks $L$ such that $p^*_\ell = 0 \;\forall\, \ell \in L$. 
    As $z^* \in \bar{\mathcal{Z}}^*(\mathbf{p}^*)$, 
    at least one \firm demands those task as much as possible because of~\cref{assumption:at-least-one-firm-demand}, and no worker wants to take those tasks since all disutility functions are strongly increasing. 
    Hence, $z^*_\ell \geq \bar{D}$ for all $\ell \in L$. 
    As $\mathbf{p}^* \in \Phi(z^*)$, because $\Phi$ is a best-response correspondence, $p^*_\ell > 0$ only if $z^*_\ell = \max_{\ell'} z^*_{\ell'} \geq \bar{D}$. 
    Consider any $\{ \{ x^*_{i\ell} \}_{i \in [n]}, \{ y^*_{j\ell} \}_{j \in [m]} \}_{\ell \in [k]}$ supporting $z^*$. 
    Note that $p^*_\ell > 0$ imples $\sum_i x^*_{i\ell} = z^*_\ell + \sum_j y^*_{j\ell} \geq z^*_\ell = \max_{\ell'} z^*_{\ell'} \geq \bar{D}$, 
    which leads to 
    \begin{equation*}
        \sum\nolimits_\ell p^*_\ell \sum\nolimits_i x^*_{i\ell} = \sum\nolimits_{\ell: p^*_\ell > 0} p^*_\ell \sum\nolimits_i x^*_{i\ell} = \bar{D} \cdot 1 > \sum\nolimits_i B_i, 
    \end{equation*} 
    which contradicts $x^*_i \in \bar{\mathcal{D}}^*_i(\mathbf{p}^*), \ \forall\, i \in [n]$ implied by $z^* \in \bar{\mathcal{Z}}^*(\mathbf{p}^*)$. 
    Note that, this contradiction means that $x^*_i \in \bar{\mathcal{D}}^*_i(\mathbf{p}^*) \;\forall\, i \in [n]$ cannot hold for \emph{any} $\{ x^*_{i\ell} \}_{i \in [n]}$ supporting $z^*$. 
    Therefore, by contradiction, $p^*_\ell > 0 \;\forall\, \ell \in [k]$. 

    For each $\ell \in [k]$, let $\{ x^*_{i\ell} \}_{i \in [n]}, \{ y^*_{j\ell} \}_{j \in [m]}$ be a pair of allocations supporting $z^*_\ell$. 
    Note that, because $\Phi$ is a best-response correspondence and $\mathbf{p}^* \in \Phi(z^*)$, 
    $p^*_\ell > 0 $ implies $\sum\nolimits_i x^*_{i\ell} - \sum\nolimits_j y^*_{j\ell} = c$ for some $c$ for all $\ell \in [k]$. 
    It follows that 
    \begin{equation}
        \sum\nolimits_\ell p^*_\ell \left( \sum\nolimits_i x^*_{i\ell} - \sum\nolimits_j y^*_{j\ell} \right) = c \sum\nolimits_\ell p^*_\ell = c. 
        \label{eq:existence-1}
    \end{equation}
    One the other hand, by $x^*_i \in \bar{\mathcal{D}}^*_i(\mathbf{p}^*)$, $y^*_j \in \mathcal{S}^*_j(\mathbf{p}^*)$ and the strong monotonicity of the disutility functions, we have 
    \begin{equation}
        \sum\nolimits_\ell p^*_\ell \left( \sum\nolimits_i x^*_{i\ell} - \sum\nolimits_j y^*_{j\ell} \right) = \sum\nolimits_i \sum\nolimits_\ell p^*_\ell x^*_{i\ell} - \sum\nolimits_j E_j  = \sum\nolimits_i \sum\nolimits_\ell p^*_\ell x^*_{i\ell} - \sum\nolimits_i B_i \leq 0, 
        \label{eq:existence-2}
    \end{equation}
    where we use $\sum_i B_i = \sum_j E_j$ in the last equality.

    Combining~\cref{eq:existence-1,eq:existence-2}, we have $\sum\nolimits_i x^*_{i\ell} - \sum\nolimits_j y^*_{j\ell} \leq 0$ for all $\ell \in [k]$. 
    With this in hand, we can show that 
    $z^*$ is also an excess demand under $p^*$ for the original ``untruncated'' market, i.e., $x^*_i = (x^*_{i1}, \ldots, x^*_{im}) \in \mathcal{D}^*_i(\mathbf{p}^*)$ for each $i$. 
    Note that $\sum_j y^*_{j\ell} \in \mathcal{B}_\infty(\bar{S})$, 
    thus we have $\sum_i x^*_{i\ell} \in \mathcal{B}_\infty(\bar{S})$. 
    Since $\bar{D} \geq 2\bar{S}$, 
    it follows that all $x^*_{i\ell} < \bar{D}$ for all $i, \ell$. 
    Next, for each $i$, we show that $x^*_i = (x^*_{i1}, \ldots, x^*_{im}) \in \mathcal{D}^*_i(\mathbf{p}^*)$ by contradiction. 
    Suppose that there is a $x'_i$ satisfying $\max_\ell x'_{i\ell} > \bar{D}$ and $u_i(x'_i) > u_i(x^*_i)$. 
    Then, one can take $(1 - \lambda) x^*_i + \lambda x'_i \in \mathcal{B}_\infty(\bar{D})$ with a sufficiently small $\lambda > 0$. 
    By the strict quasiconcavity of $u_i(\cdot)$, we have $u_i((1 - \lambda) x^*_i + \lambda x'_i) > \min\{ u_i(x^*_i), u_i(x'_i) \} = u_i(x^*_i)$, which contradicts $x^*_i \in \bar{\mathcal{D}}^*_i(\mathbf{p}^*)$. 
    
    It remains to show market clearing. 
    Since $x^*_i \in \mathcal{D}^*_i(\mathbf{p}^*)$ for each $i$ and all utility functions are locally non-satiated, the inequality in~\cref{eq:existence-2} holds with equality. 
    Market clearing then follows from $p^*_\ell > 0$ and $\sum\nolimits_i x^*_{i\ell} - \sum\nolimits_j y^*_{j\ell} \leq 0$ for all $\ell \in [k]$. 
    This completes the proof. 
\end{proof}

Note that the standard proof technique used in other economic models does not directly apply here, primarily due to the \emph{unboundedness} of the set of the feasible allocations, which undermines the compactness assumptions typically required for standard arguments. 

A direct corollary of the above theorem, together with the first welfare theorem, is the existence of a Pareto optimal allocation under the same assumptions. 
\begin{corollary}
    If the labor market satisfies~\cref{assumption:general-utility,assumption:at-least-one-firm-demand}, then there exists a nonzero Pareto optimal allocation in the labor market. 
\end{corollary}

\subsection{Convexity of Equilibrium Prices}
\label{app:prop-convex}

We next show that all equilibrium prices forms a convex cone in $\mathbb{R}_+^k$, excluding the trivial point $\p = 0$. Notice that in the one-sided setting, convexity of prices is only true for the goods case~\citep{cole2017convex}, but not for the chores~\citep{chaudhury2024competitive}. Also, in both one-sided settings the set of equilibrium prices is \textit{bounded} and does not have a conic structure like ours. A key reason is that, in our market, both supply and demand are endogenous, and scaling the equilibrium prices uniformly will cause the supply and demand to scale uniformly in the other direction, while still maintaining an equilibrium. 

\begin{theorem}
\label{thm: convexity-of-prices}
    For any input, the set of all competitive equilibrium prices is a convex cone in $\mathbb{R}_+^k$ (without zero). 
\end{theorem}
\begin{proof}
    We first prove convexity. Given two different equilibrium prices $\mathbf{p}, \mathbf{p}^\prime >0$, we consider the set of ratios on task prices $\{\rho : \exists \ \ell \in [k], p_\ell^\prime / p_\ell = \rho\}.$ Suppose there are $r$ different ratios, and sort all these ratios in decreasing order $\rho^{(1)}> \rho^{(2)} > \cdots> \rho ^{(r)}$. We can then partition the tasks into groups according their ratio $p_\ell^\prime/p_\ell$. Define $T^{(s)}$ to be the group of tasks whose ratio $p_\ell^\prime/p_\ell$ is the $s$'th largest ratio $\rho^{(s)}$,
    \begin{equation*}
        T^{(s)}:= \{\ell\in[k]: p_\ell^\prime /p_\ell = \rho^{(s)}\} . 
    \end{equation*}

    Let $(\mathbf{x}, \mathbf{y}), (\mathbf{x}^\prime, \mathbf{y}^\prime)$ be any equilibirium allocatiosn under the equilibrium prices $\mathbf{p}$ and $\mathbf{p}^\prime$, respectively. For \firm $i$ and $1\leq s\leq r$, let $b_{i,(s)} = \sum_{\ell \in T^{(s)}}p_\ell x_{i\ell}, b_{i, (s)}^\prime  = \sum_{\ell \in T^{(s)}}p_\ell^\prime x_{i\ell}^\prime$ be the corresponding spending of the \firm to the $s$'th group. We claim that for each \firm $i$, she either spends entire budget $B_i$ on $T^{(1)}$ at equilibrium $(\mathbf{p}, \mathbf{x}, \mathbf{y})$, or spends $0$ on $T^{(1)}$ at equilibrium $(\mathbf{p}^\prime, \mathbf{x}^\prime, \mathbf{y}^\prime)$, i.e.,
    \begin{equation}
    \label{eq: convexity-claim}
        (B_i-b_{i, (1)})\cdot b_{i,(1)}^\prime = 0.
    \end{equation}
    
     To see \eqref{eq: convexity-claim}, assume it does not hold for the sake of contradiction, then we have $b_{i, (1)} < B_i$ and $b_{i, (1)}^\prime >0$. Notice that $b_{i, (1)}<B_i$ implies that \firm $i$'s budget is spent also on some group ${s}\neq 1$ where $b_{i, (s)}>0$. Then there exists an MBB task $\ell^{(s)} \in T^{(s)}$ such that $\ell^{(s)} \in \arg \max_{\ell\in [k]} v_{i\ell}/p_\ell$. On the other hand, $b_{i, (1)}^\prime >0$ implies that there exists an MBB task $\ell^{(1)} \in T^{(1)}$ such that $\ell^{(1)} \in \arg \max_{\ell\in [k]} v_{i\ell}/p_\ell^\prime$. We then have $v_{i\ell^{(s)}}/p_{\ell^{(s)}} \geq v_{i\ell^{(1)}}/p_{\ell^{(1)}}$ and $v_{i\ell^{(1)}}p_{\ell^{(1)}}^\prime\geq v_{i\ell^{(s)}}/p_{\ell^{(s)}}^\prime$, which together imply $\rho^{(1)}  =p_{i\ell^{(1)}}^\prime / p_{i\ell^{(1)}}\leq p_{i\ell^{(s)}}^\prime / p_{i\ell^{(s)}} = \rho^{(s)}$. This contradicts the order $\rho^{(1)}>\rho^{(s)}$.

     Similar to \eqref{eq: convexity-claim}, let $b_{j,(s)}, b_{j,(s)}^\prime$ be the total earning of worker $j$ from task group $T^{(s)}$ in the two equilibria. We have
     \begin{equation}
         \label{eq: convexity-claim-1}
         (E_j-b_{j, (1)}^\prime)\cdot b_{j,(1)} = 0.
     \end{equation}
    Notice that \eqref{eq: convexity-claim} implies $b_{i,(1)} \geq b_{i, (1)}^\prime$ for all $i$ since the total spending on any group at the equilibrium is in $[0,B_i]$. Similarly, \eqref{eq: convexity-claim-1} implies $b_{j, (1)}\leq b_{j,(1)}^\prime$ for all $j$. By market clearance, at any equilibrium the task group $T^{(1)}$ should have the amount of spending (from the users) and earning (to the workers). We have
    \begin{equation*}
        \sumi b_{i, (1)} = \sumj b_{j, (1)} \leq \sumj  b_{j, (1)}^\prime = \sumi b_{i, (1)}^\prime  \leq \sumi b_{i, (1)} .
    \end{equation*}
    We then see that all inequalities here have to be tight, which gives us $b_{i, (1)} = b_{i, (1)}^\prime , b_{j, (1)} = b_{j,(1)}^\prime$ for any $i, j, \ell$. Combined with \eqref{eq: convexity-claim}, this tells us that any \firm $i$ spends the same amount on task group $T^{(1)}$, and this amount is either $0$ or $B_i$. Similarly, any worker $j$ earns the same amount from task group $T^{(1)}$, and this amount is either $0$ or $E_j$. For those \firms that spends the entire budgets on $T^{(1)}$ and those buyers that earns that entire requirement from $T^{(1)}$, we can remove them from the market along with group $T^{(1)}$ without affecting the rest of the equilibrium. We can then apply the same analysis to $T^{(2)}$, then inductively till $T^{(r)}$. 
    \begin{claim}
        \label{claim: convexity}
        Let $\Gamma_\mathrm{f}^{(s)}:=\{i\in[n]:b_{i,(s)}>0\}$ be the subset of \firms that spend positive budgets on $T^{(s)}$. We have $\{\Gamma_\mathrm{f}^{(s)}\}_{s=1}^r$ is a partition of the \firms and $b_{i,(s)} = b_{i, (s)}^\prime = B_i$. Similarly, let $\Gamma_\mathrm{w}^{(s)}:=\{j\in[m]:b_{j,(s)}>0\}$ be the subset of workers that have positive earning from $T^{(s)}$. We have $\{\Gamma_\mathrm{w}^{(s)}\}_{s=1}^r$ is a partition of the workers and $b_{j,(s)} = b_{j, (s)}^\prime = E_j$.
    \end{claim}

    Consider the convex combination of the given equilibrium prices, $\mathbf{p}(\lambda) = (1-\lambda)\mathbf{p} + \lambda \mathbf{p}^\prime, \lambda \in [0,1]$, then $\mathbf{p} = \mathbf{p}(0)$, $\mathbf{p}^\prime = \mathbf{p}(1)$. For task $\ell \in T^{(s)}$, we have $p_{\ell}(\lambda) = p_\ell(0) \cdot (\lambda \rho^{(s)} + 1-\lambda)$. 
    We construct allocation $\mathbf{x}(\lambda), \mathbf{y}(\lambda)$ as
    \begin{equation*}
        x_{i\ell}(\lambda) = x_{i\ell}p_{\ell}(0)/p_\ell(\lambda), \ \ y_{j\ell}(\lambda) = y_{j\ell} p_{\ell}(0)/ p_\ell(\lambda).
    \end{equation*}
    It is easy to check that $(\mathbf{p}(\lambda), \mathbf{x}(\lambda), \mathbf{y}(\lambda))$ clears the market. To show that it is an equilibrium, we show \firm optimality; the proof of worker optimality will be similar. Consider any $1\leq s\leq r$ and \firm $i\in \Gamma^{(s)}_\mathrm{f}$. Notice that prices of tasks within the group $T^{(s)}$ are scaled proportionally from $\mathbf{p}$ to $\mathbf{p}^\prime$, \Cref{claim: convexity} implies that for any $\ell^{(s)}\in T^{(s)}$ such that $b_{i\ell^{(s)}}>0$, we have
    \begin{equation*}
        \ell^{(s)} \in \arg \max_{\ell \in [k]} v_{i\ell}/p_\ell(0), \ \ \ell^{(s)} \in \arg \max_{\ell \in [k]} v_{i\ell}/p_\ell(1).
    \end{equation*}
    Now we argue that $\ell^{(s)} \in \arg \max_{\ell \in [k]} v_{i\ell}/p_\ell(\lambda)$. We show that $v_{i\ell^{(s)}}/p_{\ell^{(s)}}(\lambda) \geq v_{i\ell^{(t)}}/p_{\ell^{(t)}}(\lambda)$ for any $\ell^{(t)}\in T^{(t)}, t\in[r]$. Consider two cases.
    \begin{itemize}
        \item $t\leq  s$. Since $\ell^{(s)} \in \arg \max_{\ell \in [k]} v_{i\ell}/p_\ell(0)$ we have $v_{i\ell^{(s)}}/p_{\ell^{(s)}}(0) \geq v_{i\ell^{(t)}}/p_{\ell^{(t)}}(0)$. By the order $\rho^{(t)}\geq  \rho^{(s)}$ we have
        \begin{equation*}
            \frac{v_{i\ell^{(s)}}}{p_{\ell^{(s)}}(\lambda)} =  \frac{v_{i\ell^{(s)}}}{p_{\ell^{(s)}}(0)\cdot (\lambda \rho^{(s)}+1-\lambda)}\geq \frac{v_{i\ell^{(t)}}}{p_{\ell^{(t)}}(0)\cdot (\lambda \rho^{(t)}+1-\lambda)} =  \frac{v_{i\ell^{(t)}}}{p_{\ell^{(t)}}(\lambda)}.
        \end{equation*}
        \item $t\geq  s$. Since $\ell^{(s)} \in \arg \max_{\ell \in [k]} v_{i\ell}/p_\ell(1)$ we have $v_{i\ell^{(s)}}/p_{\ell^{(s)}}(1) \geq v_{i\ell^{(t)}}/p_{\ell^{(t)}}(1)$. By the order $\rho^{(t)}\geq  \rho^{(s)}$ we have
        \begin{equation*}
            \frac{v_{i\ell^{(s)}}}{p_{\ell^{(s)}}(\lambda)} =  \frac{v_{i\ell^{(s)}}}{p_{\ell^{(s)}}(1)\cdot ((1-\lambda)/ \rho^{(s)}+\lambda)}\geq \frac{v_{i\ell^{(t)}}}{p_{\ell^{(t)}}(1)\cdot ((1-\lambda)/ \rho^{(t)}+\lambda)} =  \frac{v_{i\ell^{(t)}}}{p_{\ell^{(t)}}(\lambda)}.
        \end{equation*}
    \end{itemize}
    This shows that $(\mathbf{p}(\lambda), \mathbf{x}(\lambda), \mathbf{y}(\lambda))$ is a competitive equilibrium and thereby $\mathbf{p}(\lambda)$ is a equilibrium price for any $\lambda \in [0,1]$. To see the conic structure, notice that $(\mu \cdot\mathbf{p}, \mathbf{x}/\mu, \mathbf{y}/\mu)$ will be a competitive equilibrium for any $\mu>0$ if $(\mathbf{p}, \mathbf{x}, \mathbf{y})$ is a competitive equilibrium. 
\end{proof}

\section{An Alternative Endogenous Demand-Supply Model}
\label{app-alternate-model}

An alternative market model with endogenous demands and supplies can be described as follows. 
The behavior of the \firms is the same as in our model, while workers' objective is to maximize their earnings subject to a pre-specified disutility threshold. 
This leads to the following supply function of a price vector $p$: 
\begin{equation*}
    \mathcal{S}_j^*(\mathbf{p}) = \underset{\mathbf{y}'_j \in \mathbb{R}^k_+}{\textnormal{argmax}} \left\{ \sum\nolimits_{\ell=1}^k p_\ell y'_{j\ell} \; \textnormal{ s.t. } d_j(\mathbf{y}'_j) \leq D_j \right\}, \hspace{20pt} \forall j \in [m].
\end{equation*}
where $D_j > 0$ denote the $i$-th worker's disutility threshold. Recall that $d_j(\mathbf{y}_j) = \sum\nolimits_{\ell=1}^k c_{j\ell} y_{j\ell}$. 

In this section, we show that the CE in this setting can be captured by the following primal and dual EG-type convex programs: 
\begin{equation}
    \begin{aligned}
        \max_{\x, \y, \mathbf{u} \geq 0} \quad & \sum_{i=1}^n B_i \log{u_i} \\ 
        \textnormal{s.t.} \quad & u_i \leq \sum\nolimits_{\ell=1}^k v_{i\ell} x_{i\ell} \;\;\; \forall\, i \in [n] & [\beta_i] \\ 
        & \sum\nolimits_{\ell = 1}^k c_{j\ell} y_{j\ell} \leq D_j \;\;\; \forall\, j \in [m] & [\alpha_j] \\ 
        & \sum\nolimits_{i=1}^n x_{i\ell} \leq \sum\nolimits_{j=1}^m y_{j\ell} \;\;\; \forall\, \ell \in [k]. & [p_\ell] 
    \end{aligned}
    \label{pgm:alter-primal}
\end{equation}
and 
\begin{equation}
    \begin{aligned}
        \min_{\mathbf{\bm \alpha}, \mathbf{\bm \beta}, \mathbf{p} \geq 0} \quad & - \sum_{i=1}^n B_i \log{\beta_i} + \sum_{j=1}^m D_j \alpha_j \\ 
        \textnormal{s.t.} \quad & v_{i\ell} \beta_i - p_\ell \leq 0 \;\;\; \forall\, i \in [n], \ell \in [k] & [x_{i\ell}] \\ 
        & - c_{j\ell} \alpha_j + p_\ell \leq 0 \;\;\; \forall\, j \in [m], \ell \in [k]. & [y_{j\ell}] 
    \end{aligned}
    \label{pgm:alter-dual}
\end{equation}


The KKT conditions of (\ref{pgm:alter-primal}) are 
\begin{enumerate}
    \item $- \frac{B_i}{u_i} + \beta_i \geq 0$, $u_i \geq 0$, and $u_i \left( - \frac{B_i}{u_i} + \beta_i \right) = 0$ for all $i \in [n]$; 
    \item $- u_i + \sum_\ell v_{i\ell} x_{i\ell} \geq 0$, $\beta_i \geq 0$, and $\beta_i \left( - u_i + \sum_\ell v_{i\ell} x_{i\ell} \right) = 0$ for all $i \in [n]$; 
    \item $v_{i\ell} \beta_i - p_\ell \leq 0$, $x_{i\ell} \geq 0$, and $x_{i\ell} \left( v_{i\ell} \beta_i - p_\ell \right) = 0$ for all $i \in [n], \ell \in [k]$; 
    \item $D_j - \sum_\ell c_{j\ell} y_{j\ell} \geq 0$, $\alpha_j \geq 0$, and $\alpha_j \left( D_j - \sum_\ell c_{j\ell} y_{j\ell} \right) = 0$ for all $j \in [m]$; 
    \item $- c_{j\ell} \alpha_j + p_\ell \leq 0$, $y_{j\ell} \geq 0$, and $y_{j\ell} \left( - c_{j\ell} \alpha_j + p_\ell \right) = 0$ for all $j \in [m], \ell \in [k]$; 
    \item $- \sum_i x_{i\ell} + \sum_j y_{j\ell} \geq 0$, $p_\ell \geq 0$, and $p_\ell \left( - \sum_i x_{i\ell} + \sum_j y_{j\ell} \right) = 0$ for all $\ell \in [k]$.
\end{enumerate}

Symmetrically, the KKT conditions of (\ref{pgm:alter-dual}) are the same as the above conditions except replacing the first two items with 
\begin{enumerate}
    \item $-\frac{B_i}{\beta_i} + \sum_\ell v_{i\ell} x_{i\ell} \geq 0$, $\beta_i \geq 0$, and $\beta_i \left( -\frac{B_i}{\beta_i} + \sum_\ell v_{i\ell} x_{i\ell} \right) = 0$ for all $i \in [n]$. 
\end{enumerate}

Next, we show that any pair of points $(\x, \y, \mathbf{u})$ and $(\mathbf{\bm \alpha}, \mathbf{\bm \beta}, \mathbf{p})$ that satisfying the KKT conditions of (\ref{pgm:alter-primal}) corresponds to a CE of the market model considered in this section. 

\begin{lemma}
    There is a one-to-one correspondence between the set of the KKT points of~\eqref{pgm:alter-primal} and competitive equilibria of the market model considered in this section. 
\end{lemma}
\begin{proof}
    Let $(\tilde{\mathbf{x}}, \tilde{\mathbf{y}}, \tilde{\mathbf{u}})$ be a KKT point of~\eqref{pgm:alter-primal} and $(\tilde{\mathbf{\bm \alpha}}, \tilde{\mathbf{\bm \beta}}, \tilde{\mathbf{p}})$ be the corresponding dual variables. 
    That is, $(\tilde{\mathbf{x}}, \tilde{\mathbf{y}}, \tilde{\mathbf{u}})$ and $(\tilde{\mathbf{\bm \alpha}}, \tilde{\mathbf{\bm \beta}}, \tilde{\mathbf{p}})$ satisfy the conditions 1.-6. listed above. 
    We first show that $(\tilde{\mathbf{p}}, \tilde{\mathbf{x}}, \tilde{\mathbf{y}})$ forms a CE. 
    \begin{enumerate}
        \item \textbf{Market clearing.} 
        Suppose that $\tilde{p}_\ell = 0$ for some $\ell \in [k]$. 
        Then, since $v_{i\ell} \tilde{\beta}_i \leq \tilde{p}_\ell$ by condition 5., we have $\tilde{\beta}_i = 0$ for all $i \in [n]$. On the other hand, 
        \begin{align*}
            \tilde{u}_i \overset{\textnormal{by condition 1.}}{\leq} \sum_{i=1}^n \tilde{u}_i \overset{\textnormal{by condition 2.}}{\leq} 
            \sum_{i=1}^n \sum_{\ell=1}^k v_{i\ell} \tilde{x}_{i\ell} &\leq \max_{i, \ell} v_{i\ell} \sum_{\ell=1}^k \sum_{i=1}^n \tilde{x}_{i\ell} \\ 
            &\overset{\textnormal{by condition 6.}}{\leq} \max_{i, \ell} v_{i\ell} \sum_{\ell=1}^k \sum_{j=1}^m \tilde{y}_{j\ell} \\ 
            &\leq \frac{\max_{i, \ell} v_{i\ell}}{\min_{j, \ell} c_{j\ell}} 
            \sum_{j=1}^m \sum_{\ell=1}^k c_{j\ell} \tilde{y}_{j\ell} \\ 
            &\overset{\textnormal{by condition 4.}}{\leq} \frac{\max_{i, \ell} v_{i\ell}}{\min_{j, \ell} c_{j\ell}} \sum_{j=1}^m D_j < \infty.
        \end{align*}
        This then contradicts the condition 1. since $0 < B_i \leq \tilde{\beta}_i \tilde{u}_i = 0$.
        Therefore, we have $\tilde{p}_\ell > 0$ for all $\ell \in [k]$.
        By condition 6., the market clearing is guaranteed.
        \item \textbf{\Firm optimality.}
        It follows by condition 3. that $v_{i\ell} \tilde{\beta}_i = \tilde{p}_\ell > 0$ if $\tilde{x}_{i\ell} > 0$ for all $i \in [n], \ell \in [k]$.
        This implies that $\tilde{\beta}_i = \frac{\tilde{p}_\ell}{v_{i\ell}}$ if $\tilde{x}_{i\ell} > 0$.
        Since we also have $v_{i\ell'} \tilde{\beta}_i \leq \tilde{p}_{\ell'}$ for any $\ell' \in [k]$, we have $\frac{v_{i\ell}}{\tilde{p}_\ell} \geq \frac{v_{i\ell'}}{\tilde{p}_{\ell'}}$ for all $\ell' \in [k]$ if $\tilde{x}_{i\ell} > 0$. 

        On the other hand, 
        \begin{equation*}
            \sum_{\ell=1}^k \tilde{p}_{\ell} \tilde{x}_{i\ell} \overset{\textnormal{by condition 3.}}{=} \tilde{\beta}_i \sum_{\ell=1}^k v_{i\ell} \tilde{x}_{i\ell} \overset{\textnormal{by condition 2.}}{=} \tilde{\beta}_i \tilde{u}_i \overset{\textnormal{by condition 1.}}{=} B_i, \hspace{10pt} \forall\, i \in [n].
        \end{equation*}
        Combining these two parts, we prove the \firm optimality.

        \item \textbf{Worker optimality.}
        Similar to the \firm case, 
        by condition 5. we have  $c_{j\ell} \tilde{\alpha}_j = \tilde{p}_\ell$ if $\tilde{y}_{j\ell} > 0$ for all $j \in [m], \ell \in [k]$.
        This implies that $\tilde{\alpha}_j = \frac{\tilde{p}_\ell}{c_{j\ell}}$ if $\tilde{y}_{j\ell} > 0$.
        Since we also have $c_{j\ell'} \tilde{\alpha}_j \geq \tilde{p}_{\ell'}$ for any $\ell' \in [k]$, 
        we have $\frac{c_{j\ell}}{\tilde{p}_\ell} \leq \frac{c_{j\ell'}}{\tilde{p}_{\ell'}}$ for all $\ell' \in [k]$ if $\tilde{y}_{j\ell} > 0$.
        
        On the other hand, since $\tilde{p}_\ell > 0$ for all $\ell \in [k]$ and $c_{j\ell} > 0$ for all $j \in [m]$ and $\ell \in [k]$, we have $\tilde{\alpha}_j > 0$ for all $j \in [m]$. 
        By condition 4., $\sum_{\ell=1}^k c_{j\ell} \tilde{y}_{j\ell} = D_j$ for all $j \in [m]$. 
        Combining these two parts, we prove the worker optimality.
    \end{enumerate}
    Conversely, for any CE $(\p^*, \x^*, \y^*)$, one can consider a pair of primal and dual variables $(\x^*, \y^*, \mathbf{u}^*)$ and $(\mathbf{\bm \alpha}^*, \mathbf{\bm \beta}^*, \p^*)$ where $u^*_i := \sum_{\ell=1}^k v_{i\ell} x_{i\ell}^*$, $\alpha^*_j := \max_{\ell \in [k]} \frac{p^*_{\ell}}{c_{j\ell}}$ and $\beta^*_i := \min_{\ell \in [k]} \frac{p^*_{\ell}}{v_{i\ell}}$. 
    It is straightforward to verify that such a pair of primal and dual variables satisfy all KKT conditions.
    This completes the proof.
\end{proof}

\section{Missing Proofs in \Cref{sec:EG-program}}
\subsection{Proof of \Cref{thm: KKT-points}}
\thmKKTPoints*
\begin{proof}
    Let ${\bm \beta} = (\beta_i)_{i \in [n]}, {\bm \alpha} = (\alpha_j)_{j \in [m]}, \p = (p_\ell)_{\ell \in [k]}$ denote the dual variables corresponding to the first, second, and third constraints, respectively. 
    Let $(\bar{\bu}, \bar{\bd}, \bar{\x}, \bar{\y})$ be a KKT point and $(\bar{\bm \beta}, \bar{ \bm \alpha}, \bar{\p})$ be the associated dual variables. 
    By KKT conditions, we have  
    \begin{enumerate}
        \item $- \frac{B_i}{\bar{u}_i} + \bar{\beta}_i \geq 0, \; \bar{u}_i > 0, \; \bar{u}_i\left( - \frac{B_i}{\bar{u}_i} + \bar{\beta}_i \right) = 0$ for all $i \in [n]$; 
        \label{primal-KKT-con:1}
        
        \item $\frac{E_j}{\bar{d}_j} - \bar{\alpha}_j \geq 0, \; \bar{d}_j > 0, \; \bar{d}_j\left( \frac{E_j}{\bar{d}_j} - \bar{\alpha}_j \right) = 0$ for all $j \in [m]$; 
        \label{primal-KKT-con:2}
        
        \item $-\bar{\beta}_i \frac{\partial u_i(\bar{x}_i)}{\partial x_{i\ell}} + \bar{p}_\ell \geq 0, \; \bar{x}_{i\ell} \geq 0, \; \bar{x}_{i\ell} \left( -\bar{\beta}_i \frac{\partial u_i(\bar{x}_i)}{\partial x_{i\ell}} + \bar{p}_\ell \right) = 0$ for all $i \in [n]$; 
        \label{primal-KKT-con:3}
        
        \item $\bar{\alpha}_j \frac{\partial d_j(\bar{y}_j)}{\partial y_{j\ell}} - \bar{p}_\ell \geq 0, \; \bar{y}_{j\ell} \geq 0, \; \bar{y}_{j\ell} \left( \bar{\alpha}_j \frac{\partial d_j(\bar{y}_j)}{\partial y_{j\ell}} - \bar{p}_\ell \right) = 0$ for all $j \in [m]$; 
        \label{primal-KKT-con:4}
        
        \item $u_i(\bar{x}_i) - \bar{u}_i \geq 0, \; \bar{\beta}_i \geq 0, \; \bar{\beta}_i \left( u_i(\bar{x}_i) - \bar{u}_i \right) = 0$  for all $i \in [n]$; 
        \label{primal-KKT-con:5}
        
        \item $\bar{d}_j - d_j(\bar{y}_j) \geq 0, \bar{\alpha}_j \geq 0, \bar{\alpha}_j \left( \bar{d}_j - d_j(\bar{y}_j) \right) = 0$ for all $j \in [m]$; 
        \label{primal-KKT-con:6}
        
        \item $\sum_{j=1}^m \bar{y}_{j\ell} - \sum_{i=1}^n \bar{x}_{i\ell} \geq 0, \; \bar{p}_\ell \geq 0, \; \bar{p}_\ell \left( \sum_{j=1}^m \bar{y}_{j\ell} - \sum_{i=1}^n \bar{x}_{i\ell} \right) = 0$ for all $\ell \in [k]$. 
        \label{primal-KKT-con:7}
    \end{enumerate} 
    We then show $(\bar{\p}, \bar{\x}, \bar{\y})$ constitutes a CE. 
    By $\bar{u}_i > 0$ (implied by the implicit open constraint), $B_i > 0$, and~(\ref{primal-KKT-con:1}), we have $\bar{\beta}_i > 0$ for each $i \in [n]$. 
    By the third assumption and KKT condition (\ref{primal-KKT-con:3}), we have $\bar{p}_\ell > 0$ for all $\ell \in [k]$. 
    As a result, market clearing follows from KKT condition (\ref{primal-KKT-con:7}) 
    For the \firm optimality, the KKT condition (\ref{primal-KKT-con:3}) yields that 
    \begin{equation}
        \bar{\beta}_i \leq \frac{\partial u_i(\bar{x}_i) / \partial x_{i\ell}}{\bar{p}_\ell} \mbox{ for each } \ell \in [k] \mbox{ and } \bar{\beta}_i = \frac{\partial u_i(\bar{x}_i) / \partial x_{i\ell}}{\bar{p}_\ell} \mbox{ whenever } \bar{x}_{i\ell} > 0. 
        \label{eq:primal-KKT-con-firm-opt-1}
    \end{equation}
    Applying KKT condition (\ref{primal-KKT-con:3}) again, we have $\sum_\ell \bar{p}_\ell \bar{x}_{i\ell} = \bar{\beta}_i \sum_\ell \bar{x}_{i\ell} \frac{\partial u_i(\bar{x}_i)}{\partial x_{i\ell}} = \bar{\beta}_i u_i(\bar{x}_i)$, where the last inequality follows by Euler's homogeneous function theorem. 
    Combining this with Conditions~\ref{primal-KKT-con:1} and \ref{primal-KKT-con:5}, we have 
    \begin{equation}
        \sum_\ell \bar{p}_\ell \bar{x}_{i\ell} = \bar{\beta}_i u_i(\bar{x}_i) = \bar{\beta}_i \bar{u}_i = B_i. 
        \label{eq:primal-KKT-con-firm-opt-2}
    \end{equation}
    \cref{eq:primal-KKT-con-firm-opt-1,eq:primal-KKT-con-firm-opt-2} implies the \firm optimality. Analogously, the worker optimality holds at the same time. 
    This shows $(\bar{\p}, \bar{\x}, \bar{\y})$ is a CE. 
    Conversely, given a CE $(\p^*, \x^*, \y^*)$, we can define 
    \begin{equation*}
        u^*_i = u_i(x^*_i), \; \beta^*_i = {B_i}/{u^*_i} \hspace{10pt} \forall\, i \in [n], \hspace{30pt} d^*_j = d_j(y^*_j), \; \alpha^*_j = {E_j}/{d^*_j} \hspace{10pt} \forall\, j \in [m]. 
    \end{equation*}
    Then, one can verify that $(\bu^*, \bd^*, \x^*, \y^*)$ is a KKT point associated with the dual variables $({\bm \beta}^*, {\bm \alpha}^*, \p^*)$. 
\end{proof}
\subsection{Proof of \Cref{lem: Dual-KKT}}
\lemDualKKT*
\begin{proof}
    Let $ \bm \gamma$, $\bm \lambda$ be the dual variables corresponding to the first and second constraints, respectively. 
    If $(\bar{\bm \beta}, \bar{\bm \alpha}, \bar{\mathbf p})$ is a KKT point and $\bar{\bm \gamma}, \bar{\bm \lambda}$ is its associated dual variables. 
    Then, by KKT conditions, we have 
    \begin{enumerate}  
        \item $- \frac{B_i}{\bar{\beta}_i} + \sum_{\ell=1}^k v_{i\ell} \bar{\gamma}_{i\ell} \geq 0, \hspace{10pt} \bar{\beta}_i > 0, \hspace{10pt} \bar{\beta}_i \big( - \frac{B_i}{\bar{\beta}_i} + \sum_{\ell=1}^k v_{i\ell} \bar{\gamma}_{i\ell} \big) = 0$ for all $i \in [n]$; 
        \label{KKT-con-a}
        
        \item $- \frac{E_j}{\bar{\alpha}_j} + \sum_{\ell=1}^k c_{j\ell} \bar{\lambda}_{j\ell} \geq 0, \hspace{10pt} \bar{\alpha}_j > 0, \hspace{10pt} \bar{\alpha}_j \big( - \frac{E_j}{\bar{\alpha}_j} + \sum_{\ell=1}^k c_{j\ell} \bar{\lambda}_{j\ell} \big) = 0$ for all $j \in [m]$; 
        \label{KKT-con-b} 
        
        \item $-\sum_{i = 1}^n \bar{\gamma}_{i\ell} + \sum_{j = 1}^m \bar{\lambda}_{j\ell} \geq 0, \hspace{10pt} \bar{p}_\ell \geq 0, \hspace{10pt} \bar{p}_\ell \big( -\sum_{i = 1}^n \bar{\gamma}_{i\ell}  + \sum_{j = 1}^m \bar{\lambda}_{j\ell} \big) = 0$ for all $\ell \in [k]$; 
        \label{KKT-con-c}
        
        \item $v_{i \ell} \bar{\beta}_i \leq \bar{p}_\ell, \hspace{10pt} \bar{\gamma}_{i\ell} \geq 0, \hspace{10pt} \bar{\gamma}_{i\ell}\left( v_{i \ell} \bar{\beta}_i - \bar{p}_\ell \right) = 0$ for all $i \in [n], \ell \in [k]$; 
        \label{KKT-con-d}
        
        \item $\bar{p}_\ell \leq d_{j \ell} \bar{\alpha}_j, \hspace{10pt} \bar{\lambda}_{j\ell} \geq 0, \hspace{10pt} \bar{\lambda}_{j\ell}\left( \bar{p}_\ell - d_{j \ell} \bar{\alpha}_j \right) = 0$ for all $j \in [m], \ell \in [k]$. 
        \label{KKT-con-e}
    \end{enumerate}
    First, note that 
    $\bar{\beta}_i > 0 \; \forall\, i \in [n]$ implies $\bar{p}_\ell > 0 \; \forall\, \ell \in [k]$ because for every task $\ell$ there is at least one $i$ such that $v_{i\ell} > 0$ and Condition~\ref{KKT-con-d} holds.  
    By Condition~\ref{KKT-con-c}, we ensure market clearing. 
    For any buyer $i \in [n]$ and task $\ell \in [k]$, 
    if $\bar{\gamma}_{i\ell} > 0$, 
    then by Condition~\ref{KKT-con-d} we have $\frac{v_{i \ell}}{\bar{p}_\ell} = \frac{1}{\bar{\beta}_i} \geq \frac{v_{i \ell'}}{\bar{p}_{\ell'}}$ for any $\ell' \in [k]$. 
    Meanwhile, we have $\sum_{\ell' = 1}^k \bar{p}_{\ell'} \bar{\gamma}_{i \ell'} = \bar{\beta}_i \sum_{\ell' = 1}^k v_{i \ell'} \bar{\gamma}_{i\ell'} = B_i$ by Conditions~\ref{KKT-con-d} and~\ref{KKT-con-a}. 
    This guarantees the \firm optimality. 
    Similarly, we can verify the worker optimality. 
    Therefore, $(\bar{\p}, \bar{\bm \gamma}, \bar{\bm \lambda})$ constitutes a CE. 

    Conversely, given that $(\p^*, \x^*, \y^*)$ is a CE, we define 
    \begin{equation*}
        \beta^*_i = \min_{\ell \in [k]} \frac{p^*_\ell}{v_{i\ell}} \; \forall\, i \in [n] \hspace{30pt} \text{ and } \hspace{30pt} \alpha^*_j = \max_{\ell \in [k]} \frac{p^*_\ell}{d_{i\ell}} \; \forall\, j \in [m]. 
    \end{equation*}
    Let $(\mathbf \p^*, \bm \beta^*, \bm \alpha^*)$ and $(\x^*, \y^*)$ be the primal and dual variables in~(\ref{program:EG-Dual}), one can verify that they satisfy the KKT conditions, and thus $(\p^*, \bm \beta^*, \bm \alpha^*)$ is a KKT point coupled with $(\x^*, \y^*)$. 
    
\end{proof}

\section{Missing Proofs in \Cref{sec:walrasian-algo}}

\subsection{Proof of \Cref{lem:price-update}}
\lemPriceUpdate*
\begin{proof}
We consider partitioning the set of \firms and workers as follows:
\begin{equation*}
    \begin{aligned}
        & \Gamma_{\mathrm{f}}(S) = \{i\in[n]: \exists \ \ell \in S, \ b_{i\ell} >0\},
        \ \ \ \overline{\Gamma}_{\mathrm{f}}(S) = [n]\backslash \Gamma_{\mathrm{f}}(S) , \\ 
        & \Gamma_{\mathrm{w}}(S) = \{j\in[m]: \exists \ \ell^\prime \notin S, \ e_{j\ell^\prime} >0\}, \ \ \ 
        \overline{\Gamma}_{\mathrm{w}}(S) = [m]\backslash \Gamma_{\mathrm{w}}(S).
    \end{aligned}
\end{equation*}
Here, $\Gamma_{\mathrm{f}}(S)$ includes all \firms that spend positive amount in some task in $S$, and $\Gamma_{\mathrm{w}}(S)$ includes all workers that earn positive amount in some task in $S$. Consider the following cases,

\begin{itemize}
    \item  $i \in \Gamma_\mathrm{f}(S)$. 
    In this case, we must have $\arc{i\ell'} \notin \mathcal{E}$ for all $\ell' \notin S$, i.e., \firm $i$ has no MBB task outside $S$. Suppose this is not the case, then it is possible to 
    move a nonzero amount of flow from $b_{i\ell}$ to $b_{i\ell'}$. 
    However, this implies that the previous spending update phase would either have more space to adjust the flow (so the update phase should have continued), or reach the lower bound $\lb$ on some surpluses (so the algorithm proceeds to the next iteration). In either cases we would not have entered the price update procedure, which gives a conflict. 

    Now that we know any edge incident on $i$ in $\mathcal{G}(\p)$ must have the other end in $S$, it then follows that none of these edges will disappear in the price update, since the prices of tasks in $S$ are scaled up in proportion. 

    \item $j \in {\Gamma}_{\mathrm{w}}(S)$. Similar to the analysis of ${\Gamma}_{\mathrm{f}}(S)$, we know that any edge incident on $j$ in $\mathcal{G}(\p)$ must have the other end in $[k]\backslash S$. Since the prices of tasks in $[k]\backslash S$ remain unchanged, none of these edges will disappear in the price update.
     
    \item $i\in \overline{\Gamma}_{\mathrm{f}}(S)$. Since we only increase the price of tasks in $S$, an edge $\arc{i\ell}$ disappears only if $\ell \in S$. By the definition of the partition, such edge carries zero flow in $\mathcal{G}(\p)$.
    
    \item $j\in \overline{\Gamma}_{\mathrm{w}}(S)$. Since we only increase the price of tasks in $S$, an edge $\arc{\ell^\prime j}$ disappears only if $\ell^\prime \in [k]\backslash S$. By the definition of the partition, such edge carries zero flow in $\mathcal{G}(\p)$. 
\end{itemize}
Combining the above cases, we have that no edge with positive flow disappears from the MBB/MPB graph during the price update. 
\end{proof}

\subsection{Proof of \Cref{lem:allocation-update}}
\lemAllocationUpdate*
\begin{proof}
    We first argue that each phase terminates after polynomial number of operations. For the phases of spending and earning update this is clear. For the price update, notice that each edge can only appear or disappear once in the MBB/MPB set. Let $\gamma_i(S,\p), \gamma_i([k]\backslash S,\p)$ denote the MBB value of \firm $i$ over the set of tasks $S$ and $[k]\backslash S$, respectively, and let $\bar{\p}$ be the prices in the beginning of an iteration. By proportionality of the price update, we know it must be one of the following two cases:
    \begin{enumerate}
        \item $\gamma_i(S, \bar{\p})> \gamma_i([k]\backslash S, \bar{\p})>0$, then all possible appearance of edges that are incident on $i$ will happen when $\p$ is scaled to $a_i^* \cdot \bar{\p}$ where $a_i^*= \gamma_i(S, \bar{\p})/\gamma_i([k]\backslash S, \bar{\p})$.
        \item Otherwise, no appearance or disappearance will happen for any edges that are incident on $i$.
    \end{enumerate}
    Therefore, at the beginning of each iteration, we can check the above two cases for each \firm $i$ in advance, and compute $a_i^*$ if we are in the first case. For price updates, we only need to let the scaling multiplier be $a_i^*$'s in increasing order. This takes polynomial time, and is to be combined with similar multipliers $a_j^*$'s from the worker side, which can be computed in advance similarly.

    For the number of price update phases, we know from above that there will be at most $n+m$ of such updates that cause the appearance of new MBB/MPB edges. However, there is still the risk of entering the price update phase again after all possible edge appearances, and the algorithm ends up scaling prices to infinity without terminating. We argue that this never happens: When prices in $S$ are scaled uniformly by a sufficiently large multiplier, there will be no edges between $[k]\backslash S$ and the workers, so the earning updates of all buyers would have shifted all earning requirements to $S$. This implies at least one $\ell \in S$ with non-positive surplus, which contradicts the fact that no surplus in $S$ is below $\lb>0$ in an iteration. 
    
    We then know that the number of price update phases is at most $n+m$. For the spending and earning update phases, notice that each either hits a lower bound $\lb$ (which ends the iteration) or is followed by a price update. The number of spending and earning update phases in an iteration is then upper-bounded by the number of price update phases plus $1$. 
\end{proof}
\section{Missing Discussions on the LP Formulation} 
\label{appendix: discussion-LP}
\subsection{Equivalence to a Min-cost Flow Problem}
Consider changing the variables of \eqref{program:EG Dual log Dual}:
let $f_{i\ell j}$ denote the amount of the money flow from $i$ to $j$ via $\ell$ for all $i \in [n], \ell \in L(i), j \in [m]$. 
Then, $b_{i\ell} = \sum_{j} f_{i\ell j}$, $e_{j\ell} = \sum_{i: \ell \in L(i)} f_{i\ell j}$.

\eqref{program:EG Dual log Dual} is equivalent to 
\begin{equation}
    \begin{aligned}
        \min_{f \geq 0} \quad & \sum_{i} \sum_{\ell \in L(i)} \sum_j (-\log{v_{i\ell}}) f_{i\ell j} + \sum_{j} \sum_{\ell} \sum_{i: \ell \in L(i)} (\log{c_{j\ell}}) f_{i\ell j} \\ 
        \textnormal{s.t.} \quad 
        & \sum_{\ell \in L(i)} \sum_{j} f_{i\ell j} = B_i \quad \forall\, i \in [n] \\ 
        & \sum_{\ell} \sum_{i: \ell \in L(i)} f_{i\ell j} = E_j \quad \forall\, j \in [m] \\ 
        & \sum_{i: \ell \in L(i)} \sum_{j} f_{i\ell j} = \sum_j \sum_{i: \ell \in L(i)} f_{i\ell j} \quad \forall\, \ell \in [k].  
    \end{aligned}
    \label{program:EG Dual log Dual replace}
\end{equation}
Equivalently, 
\begin{equation}
    \begin{aligned}
        \min_{f \geq 0} \quad & \sum_{i} \sum_j \sum_{\ell \in L(i)} (\log{c_{j \ell}} - \log{v_{i\ell}}) f_{i\ell j} \\ 
        \textnormal{s.t.} \quad 
        & \sum_{j} \sum_{\ell \in L(i)} f_{i\ell j} = B_i \quad \forall\, i \in [n] \\ 
        & \sum_{i} \sum_{\ell \in L(i)} f_{i\ell j} = E_j \quad \forall\, j \in [m].  
    \end{aligned}
    \label{program:EG Dual log Dual replace equiv}
\end{equation} 

The above LP corresponds to a min-cost flow problem on a multigraph consisting of a set of vertices representing users and workers, and $k$ edges between each pair of user-worker vertices. 
Since the cost-minimizing flow only goes through the edge with the minimum cost, the optimal solution to~\eqref{program:EG Dual log Dual replace equiv} can be captured by 
\begin{equation}
    \begin{aligned}
        \min_{a \geq 0} \quad & \sum_{i} \sum_{j} w_{i j} a_{i j} \\ 
        \textnormal{s.t.} \quad 
        & \sum_{j} a_{i j} = B_i \quad \forall\, i \in [n] \\ 
        & \sum_{i} a_{i j} = E_j \quad \forall\, j \in [m],  
    \end{aligned}
    \tag{Labor Market LP'}
    \label{pgm:simple-min-cost}
\end{equation}
where $w_{i j} := \min_{\ell' \in L(i)} (\log{c_{j\ell'}} - \log{v_{i\ell'}})$. 

\subsection{Discussions on Using a Truncated LP}
\label{app:subsec:truncated-LP}
To ensure that a small enough perturbation in the objective coefficients does not change the optimal solution, we need to guarantee that the objective value increases by a certain gap whenever the solution deviates from the optimal point to a non-optimal vertex of the feasible region. 

In other words, we need to compute the minimal objective change between an optimal vertex and a non-optimal vertex. 
However, this requires us to compute the modulus of weak sharp minima of~(\ref{program:EG Dual log Dual}). 
In the context of LPs, the modulus of weak sharp minima serves as an exact error‐bound constant, and—when enforcing polyhedrality—it coincides with the Hoffman constant. 
This connection has been 
established in \citet{burke2005weak}. 
It is known that there is no efficient way to compute Hoffman or error bound constants in general. 
In particular, computing exact Hoffman constants is known to be NP-hard~\citep{pena2021new}. 

Note that, we need to compute such a constant even before solving the LP. 
Intuitively, we need to consider every pair of possible vertices of the feasible region, or equivalently every pair of basic feasible solutions, which can be exponentially many.

\section{Missing Proofs in \Cref{sec:LP}}
\label{app:subsec:proofs-LP}

\subsection{Proof of \Cref{thm: LP-Equivalence}}
\thmLPEquivalence*
\begin{proof}
    Let $(\gamma_i)_{i \in [n]}, (\lambda_j)_{j \in [m]}, (\mu_\ell)_{\ell \in [k]}$ denote the dual variables corresponding to the first, second, and third constraints in~(\ref{program:EG Dual log Dual}), respectively. 
    Let $(\bar{\mathbf b}, \bar{\mathbf e})$ be a KKT point of~(\ref{program:EG Dual log Dual}) and $(\bar{\bm \gamma}, \bar{\bm\lambda}, \bar{\bm\mu})$ be the associated dual variables, then the KKT conditions lead to 
    \begin{enumerate}
        \item $- \log{v_{i\ell}} + \bar{\gamma}_i + \bar{\mu}_\ell \geq 0, \hspace{10pt} \bar{b}_{i\ell} \geq 0, \hspace{10pt} \bar{b}_{i\ell} \left( - \log{v_{i\ell}} + \bar{\gamma}_i + \bar{\mu}_\ell \right) = 0$ for all $i \in [n], \ell \in L(i)$; 
        \label{LP-KKT-con-1} 
        
        \item $\log{c_{j\ell}} + \bar{\lambda}_j - \bar{\mu}_\ell \geq 0, \hspace{10pt} \bar{e}_{j\ell} \geq 0, \hspace{10pt} \bar{e}_{j\ell} \left( \log{c_{j\ell}} + \bar{\lambda}_j - \bar{\mu}_\ell \right) = 0$ for all $j \in [m], \ell \in [k]$; 
        \label{LP-KKT-con-2}

        \item $\sum_{\ell \in L(i)} \bar{b}_{i\ell} = B_i$ for all $i \in [n]$; 
        \label{LP-KKT-con-3}

        \item $\sum_\ell \bar{e}_{j\ell} = E_j$ for all $j \in [m]$; 
        \label{LP-KKT-con-4}

        \item $\sum_{i: \ell \in L(i)} \bar{b}_{i\ell} = \sum_j \bar{e}_{j \ell}$ for all $\ell \in [k]$. 
        \label{LP-KKT-con-5}
    \end{enumerate}
    Let $\bar{p}_\ell = e^{\bar{\mu}_\ell} > 0$ for all $\ell \in [k]$, $\bar{x}_{i\ell} = {\bar{b}_{i\ell}}/{\bar{p}_\ell}$ for all $i \in [n], \ell \in [k]$ 
    and $\bar{y}_{j \ell} = {\bar{a}_{j\ell}}/{\bar{p}_\ell}$ for all $j \in [m], \ell \in [k]$. 
    Market clearing is ensured by Condition~\ref{LP-KKT-con-5}. 
    Also, $\sum_\ell \bar{p}_\ell \bar{x}_{i\ell} = \sum_\ell \bar{b}_{i\ell} = B_i$ by Condition~\ref{LP-KKT-con-3} and 
    $\sum_\ell \bar{p}_\ell \bar{y}_{j\ell} = \sum_\ell \bar{e}_{j\ell} = E_j$ by Condition~\ref{LP-KKT-con-4}. 
    It remains to show that 1) $v_{i\ell} / \bar{p}_\ell = \max_{\ell'} v_{i\ell'} / \bar{p}_{\ell'}$ whenever $\bar{x}_{i\ell} > 0$ and 2) $\bar{p}_\ell / c_{j\ell} = \max_{\ell'} \bar{p}_{\ell'} / c_{j\ell'}$ whenever $\bar{y}_{j\ell} > 0$. 
    These properties are implied by the KKT conditions: $\bar{x}_{i\ell} > 0$ implies $\bar{b}_{i\ell} > 0$, and thus $\log{v_{i\ell}} - \log{\bar{p}_\ell} = \gamma_i$ by Condition~\ref{LP-KKT-con-1}. Then, the first inequality in Condition~\ref{LP-KKT-con-1} implies $\log{v_{i\ell}} - \log{\bar{p}_\ell} \geq \log{v_{i\ell'}} - \log{\bar{p}_{\ell'}}$ for all $\ell' \in [k]$. 
    Similarly, we can show that every worker only chooses her MPB tasks. 
    Therefore, $(\bar{\p}, \bar{\x}, \bar{\y})$ is a CE. 

    Conversely, if $(\p^*, \x^*, \y^*)$ is a CE, then we can define 
    \begin{align*}
        b^*_{i\ell} = x^*_{i\ell} p^*_\ell \hspace{10pt} e^*_{j\ell} = y^*_{j\ell} p^*_\ell \hspace{10pt} \gamma_i = \max_{\ell \in [k]} \left(\log{\frac{v_{i\ell}}{p^*_\ell}}\right) \hspace{10pt} \lambda_j = \max_{\ell \in [k]} \left(\log\frac{p^*_\ell}{c_{j\ell}}\right) &\hspace{10pt} \mu_\ell = \log{p^*_\ell} \hspace{30pt} \nonumber \\ &\forall\, i \in [n], j \in [m], \ell \in [k]. 
    \end{align*}
    One can then verify that $(\mathbf b^*, \mathbf e^*)$ is a KKT point associated with the dual variables $(\bm \gamma^*, \bm\lambda^*, \bm\mu^*)$. 
\end{proof}

\subsection{Proof of \Cref{claim:mult-approx}}
\label{app:sec:truncate-coefficients-LP}
\claimMultApprox*
\begin{proof}
    Because we are rounding down, $v_{i\ell} \geq \tilde{v}_{i\ell}$ and $c_{j\ell} \geq \tilde{c}_{j\ell}$. 
    On the other hand, 
    since they are rounded up to $L$ decimal bits, we have $|\log{\tilde{v}_{i \ell}} - \log{v_{i \ell}}| \leq 2^{-L}$, 
    and $|\log{\tilde{c}_{j \ell}} - \log{c_{j \ell}}| \leq 2^{-L}$. 
    The lower bound of $\tilde{v}_{i\ell}$ follows by 
    $\tilde{v}_{i\ell} = e^{\log{\tilde{v}_{i\ell}}} \geq e^{\log{v_{i\ell}}} \cdot e^{-2^{-L}} \geq v_{i\ell} (1+2\cdot2^{-L})^{-1}$ where the last inequality holds because $e^{-x} \geq \frac{1}{1+2x}$ for $x \in [0, 1]$. 
    Similarly, the lower bound for $\tilde{c}_{j\ell}$ holds as well. 
\end{proof}


\subsection{Proof of \Cref{lem: graph-of-adjusted-prices}}
\lemAdjusted*
\begin{proof}
    Notice that there can be multiple connected components in $\tilde{\mathcal{G}}(\tilde{\mathbf{p}})$. 
    Let $\psi_1, \psi_2, \dots, \psi_r$ be the connected components of $\tilde{\mathcal{G}}(\tilde{\mathbf{p}})$. We start with an arbitrary component, say $\psi_1$. 
    Then, we multiplicatively increase the prices of all tasks in $\psi_1$, until 
    (i) a new MBB edge appears from some \firm in $\psi_1$ to a task outside of $\psi_1$, \emph{or} (ii) a new MPB edge appears from some worker outside of $\psi_1$ to a task in $\psi_1$. 
    In either of the two cases, 
    $(1)$ we still maintain that the price is an equilibrium price coupled with $(\tilde{\mathbf{b}}, \tilde{\mathbf{e}})$
    because \emph{there is no MBB/MPB edge disappearing} (thus the spendings and earnings flow only along MBB/MPB edges); and 
    $(2)$ the number of connected components decreases. 
    We repeat this process until 
    $\tilde{\mathcal{G}}(\tilde{\mathbf{p}})$ is connected. 
    Again, 
    since scaling price preserves it as an equilibrium price, 
    we can normalize the price vector such that $\sum_{\ell \in [k]} \tilde{p}_{\ell} =1$. 
\end{proof}

\subsection{Proof of~\Cref{lem:find-an-eq-price-in-polytime}}

\lemFindAnEqPriceInPolytime*

\begin{proof}
    To find such a price, we solve a system of equations which is defined by paths in $\tilde{\mathcal{G}}(\tilde{\mathbf{p}})$. 
    Fix the $k$-th task (the one with the largest index) as the representative task. 
    Since $\tilde{\mathcal{G}}(\tilde{\mathbf{p}})$ is connected, 
    for any task $\ell \in [k - 1]$, there is a path involving $\tau_\ell$ agents: 
    $q_{1}(\ell) - i_1(\ell) - q_2(\ell) - i_2(\ell) - \cdots  - i_{\tau_\ell}(\ell) - q_{\tau_\ell+1}(\ell)$ from $q_1(\ell) = k$ to ${q}_{\tau_\ell+1}(\ell) = \ell$ in $\tilde{\mathcal{G}}(\tilde{\mathbf{p}})$,
    where $q_r(\ell)$ denotes the $r$-th task on the path, and $i_r(\ell)$ denotes the $r$-th agent on the path. 
    The path depends on $\ell$ has length $2 \tau_{\ell}$.

    We use $\bm{\nu}$ to encode input utility data.
    Given input utility data $\bm{\nu}$, for each $\ell \in [k - 1]$, we define 
    $t_{\ell}(\bm{\nu}) := \prod_{r \in [\tau_\ell]} \left({\nu_{i_r(\ell) {q}_{r+1}(\ell)}}\big/{\nu_{i_r(\ell) {q}_r(\ell)}}\right)$ 
    where $\nu_{i {q}} = v_{i {q}}$ if $i$ is a user, and $\nu_{i q} = c_{i {q}}$ if $i$ is a worker. We then define the following system of linear equations:
    \begin{equation}
        \begin{aligned}
        p_{\ell} - t_{\ell}(\bm{\nu}) p_k &= 0, \quad \forall\, \ell \in [k - 1] \\ 
        \sum\nolimits_{\ell \in [k]}p_\ell &= 1.
        \end{aligned}
        \label{MPB-system-equations}
    \end{equation} 

    For given input utility data $\bm{\nu}$, we write the linear system \eqref{MPB-system-equations} as $A(\bm{\nu}) \p = \mathbf{z}$.
    Then, we can obtain a price $\p$ by solving $A(\bm{\nu}) \p = \mathbf{z}$.

    \begin{restatable}{claim}{claimMPB}
        \label{claim: MPB-system-equations-solution-uniqueness}
         $A(\bm{\nu}) \p = \mathbf{z}$ admits a unique solution. 
    \end{restatable}
    \begin{proof}
    
    Since $A(\bm{\nu})$ is invertible with the inverse given by
    \begin{equation*}
        A^{-1}(\bm{\nu}) = \frac{1}{s(\bm{\nu})}\begin{pmatrix}
            s(\bm{\nu}) - t_1(\bm{\nu}) & - t_1(\bm{\nu}) & \cdots & - t_1(\bm{\nu}) & t_1(\bm{\nu}) \\ 
            - t_2(\bm{\nu}) & s(\bm{\nu}) - t_2(\bm{\nu}) & \cdots & - t_2(\bm{\nu}) & t_2(\bm{\nu}) \\ 
            \vdots & \vdots & \ddots & \vdots & \vdots \\ 
            - t_{k - 1}(\bm{\nu}) & - t_{k - 1}(\bm{\nu}) & \cdots & s(\bm{\nu}) - t_{k - 1}(\bm{\nu}) & t_{k - 1}(\bm{\nu}) \\ 
            -1 & -1 & \cdots & -1 & 1
        \end{pmatrix}, 
    \end{equation*} 
    where $s(\bm{\nu}) = 1 + \sum_{\ell = 1}^{k - 1} t_\ell(\bm{\nu}) > 0$, the system of linear equations $A(\bm{\nu}) \mathbf p = \mathbf{z}$ admits the following unique solution
    \begin{equation}
        \mathbf{p} = \frac{1}{s(\bm{\nu})}\Big( t_1(\bm{\nu}), t_2(\bm{\nu}), \ldots, t_{k - 1}(\bm{\nu}), 1 \Big).
        \label{eq:equilibrium-price-system-equations-unique-solution}
    \end{equation}

\end{proof}
    
    
    Let $\mathbf{u}$ and $\tilde{\mathbf{u}}$ denote the input utility data of the original and rounded instances, respectively. 
    We see that $\tilde{\mathbf{p}}$ satisfies $A(\tilde{\mathbf{u}}) \tilde{\mathbf{p}} = \mathbf{z}$ by the definition of the MBB/MPB graph and $\sum_\ell p_\ell = 1$.

    To find the desired price vector, we solve \eqref{MPB-system-equations} with $\bm{\nu} = \bf u$; denote its unique solution as $\mathbf{p}^*$. The system of linear equations is of polynomial-size and can be solved in $\mathcal{O}(k^3)$ arithmetic operations. 
    The following claim shows a lower bound of the minimal difference between two coordinates of $\mathbf{p}^*$. 
    \begin{restatable}{claim}{claimPBound}
    \label{claim: p-coordinate-bound}
        If $p_\ell^* \neq p_{\ell'}^*$ for some $\ell, \ell' \in [k]$, then we have $|p_{\ell}^* - p_{\ell'}^*| \geq k^{-1} (U D)^{-k^2}$.
        \label{fact:rationality-minimal-gap}
    \end{restatable}
    \begin{proof}
        For each $\ell \in [k-1]$, denote $a_\ell := \prod_{r \in [\tau_\ell]} u_{i_r(\ell) {q}_{r+1}(\ell)}$ and $b_\ell := \prod_{r \in [\tau_\ell]} u_{i_r(\ell) {q}_r(\ell)}$ in the definition of $t_\ell(\mathbf{u})$. 
        Since $\mathbf{u}$ is the input utility data of the original instance thus rational, 
        $t_\ell(\bm{\nu}) = a_{\ell}/b_{\ell}$
        where both $a_{\ell}$ and $b_{\ell}$ are positive integers and $\max\{ a_{\ell}, b_{\ell} \} \leq (U D)^k$ for all $\ell \in [k-1]$. 
        By~\cref{eq:equilibrium-price-system-equations-unique-solution}, 
        any $p^*_\ell$ can be expressed as 
        \begin{equation*}
            \frac{a_\ell / b_\ell}{1 + \sum_{\ell' = 1}^{k - 1} a_{\ell'} / b_{\ell'}} = \frac{a_\ell\Pi_{q \in [k - 1] \setminus \{ \ell \}} b_q}{\Pi_{q \in [k - 1]} b_q + \sum_{\ell' = 1}^{k - 1} a_{\ell'}\Pi_{q \in [k - 1] \setminus \{ \ell' \}} b_q}. 
        \end{equation*}
        Therefore, the rationality dictates that the gap between two different prices is at least $\frac{1}{k (U D)^{k^2}}$. 
    \end{proof}

    Moreover, we present the following claim to upper bound the difference between $\p^*$ and $\tilde{\p}$.
    \begin{claim}
             \label{claim:p&tildep}
             For all $\ell \in [k]$, we have $|p^*_{\ell} - \tilde{p}_{\ell}| \leq 24k \cdot 2^{-L}$. 
        \end{claim}
    \begin{proof}
        We use $i_r$ and $q_r$ to replace $i_r(\ell)$ and $q_r(\ell)$ for simplicity.
        For any $\ell \in [k]$, 
        we have 
        \begin{equation}
            \frac{1}{(1 + 2^{1-L})^k} 
            \leq \prod_{r \in [\tau_\ell]} \frac{\tilde{u}_{i_r {q}_{r+1}}}{u_{i_r {q}_{r+1}}} 
            \leq \frac{t_\ell(\tilde{\mathbf{u}})}{t_\ell(\mathbf{u})} 
            = \prod_{r \in [\tau_\ell]} \frac{\tilde{u}_{i_r {q}_{r+1}} u_{i_r {q}_r}}{\tilde{u}_{i_r {q}_r} u_{i_r {q}_{r+1}}} 
            \leq \prod_{r \in [\tau_\ell]} \frac{u_{i_r {q}_r}}{\tilde{u}_{i_r {q}_r}} 
            \leq (1 + 2^{1-L})^k. 
            \label{eq:ratio-t-bound}
        \end{equation}

        Define $t_k(\bm{\nu}) = 1$ for any input utility data $\bm{\nu}$. 
        By~\cref{eq:equilibrium-price-system-equations-unique-solution}, 
        for any $\ell \in [k]$, 
        we have 
        \begin{align*}
            |p^*_{\ell} - \tilde{p}_{\ell}| = \left| \frac{t_\ell(\mathbf{u})}{s(\mathbf{u})} - \frac{t_\ell(\tilde{\mathbf{u}})}{s(\tilde{\mathbf{u}})} \right| 
            &= \frac{t_\ell(\tilde{\mathbf{u}})}{t_\ell(\mathbf{u})}\frac{t_\ell(\mathbf{u})}{s(\mathbf{u})} \left| \frac{t_\ell(\mathbf{u})}{t_\ell(\tilde{\mathbf{u}})} - \frac{s(\mathbf{u})}{s(\tilde{\mathbf{u}})} \right| \\  
            &\overset{(a)}{\leq} \left( 1 + 2^{1-L} \right)^k \left(\max_{\ell \in [k]} \frac{t_\ell(\mathbf{u})}{t_\ell(\tilde{\mathbf{u}})} - \min_{\ell \in [k]} \frac{t_\ell(\mathbf{u})}{t_\ell(\tilde{\mathbf{u}})} \right) \\
            &\overset{(b)}{\leq} \left( 1 + 2^{1-L} \right)^k \left( \left( 1 + 2^{1-L} \right)^k - (1 + 2^{1-L})^{-k} \right), 
        \end{align*}
        where $(a)$ follows by~\cref{eq:ratio-t-bound}, $\frac{t_\ell(\mathbf{u})}{s(\mathbf{u})} \leq 1$ and $\min_{\ell \in [k]} \frac{t_\ell(\mathbf{u})}{t_\ell(\tilde{\mathbf{u}})} \leq \frac{s(\mathbf{u})}{s(\tilde{\mathbf{u}})} \leq \max_{\ell \in [k]} \frac{t_\ell(\mathbf{u})}{t_\ell(\tilde{\mathbf{u}})}$, 
        and $(b)$ follows by~\cref{eq:ratio-t-bound}. 
        Recall that $L \geq k^2$ and thus $0 < 2^{1 - L} \leq 2^{1 - k^2} \leq \frac{1}{k}$ for $k \geq 1$. 
        Because $(1 + x)^k \leq e^{xk} \leq 2kx + 1$ and $(1 + x)^{-k} \geq e^{-kx} \geq -2kx + 1$ for any $0 < x \leq \frac{1}{k}$, we have $|p^*_{\ell} - \tilde{p}_{\ell}| \leq (2k \cdot 2^{1-L} + 1) \cdot 4k \cdot 2^{1-L} \leq 12k \cdot 2^{1-L} = 24k \cdot 2^{-L}$ (we use $k \cdot 2^{1-L} \leq 1$ in the last inequality). 
\end{proof}

To show the statement in~\cref{lem:find-an-eq-price-in-polytime}, it suffices to show the following lemma.
\begin{restatable}{lemma}{lemMBBExtraction}
     \label{lem:MBB-extraction}
     If $\arc{i\ell} \in \tilde{\mathcal{E}}(\tilde{\mathbf{p}})$, then $\arc{i\ell} \in \mathcal{E}(\mathbf{p}^*)$, and symmetrically, $\arc{\ell j} \in \tilde{\mathcal{E}}(\tilde{\mathbf{p}})$, then $\arc{\ell j} \in \mathcal{E}(\mathbf{p}^*)$.
\end{restatable}
\begin{proof}
    We prove this lemma by contradiction. 
    Suppose that, there exists $\arc{i\ell} \in \tilde{\mathcal{E}}(\tilde{\mathbf{p}})$, but $\arc{i\ell} \notin \mathcal{E}(\mathbf{p}^*)$, i.e., there exists an $\ell' \in [k]$ such that $v_{i \ell}/p^*_{\ell} < v_{i \ell'} / p^*_{\ell'}$, 
    or equivalently, $v_{i \ell} p^*_{\ell'} < v_{i \ell'} p^*_{\ell}$. 
    By \Cref{fact:rationality-minimal-gap} and $(v_{i\ell})_{i \in [n], \ell \in[k]}$ are integers, we have 
    \begin{equation}
        v_{i \ell} p_{\ell'} \leq v_{i \ell'} p_{\ell} - k^{-1} (U D)^{-k^2}. 
        \label{eq:rational-gap}
    \end{equation} 
    Then, since $L = k^2(\log_2{(U D)}) + \log_2{U} + \log_2{D} + 2\log_2{k} + 8$, we have 
    \begin{align}
        \tilde{v}_{i \ell} \tilde{p}_{\ell'} 
            &\leq v_{i \ell} (p_{\ell'}^* + 24k \cdot 2^{-L}) \tag{by~\cref{claim:p&tildep} and $\tilde{v}_{i\ell} \leq v_{i\ell}$} \\
            &\leq v_{i \ell} p_{\ell'}^* + 24k U \cdot 2^{-L} \nonumber \\
            &\leq v_{i \ell'} p_{\ell}^* + 24k U \cdot 2^{-L} - k^{-1} (U D)^{-k^2} \tag{by~\cref{eq:rational-gap}} \\ 
            &\leq (\tilde{v}_{i \ell'} + U \cdot 2^{1-L})(\tilde{p}_{\ell} + 24k \cdot 2^{-L}) + 24k U \cdot 2^{-L} - k^{-1} (U D)^{-k^2} \tag{by~\cref{claim:p&tildep} and $v_{i\ell} \leq (1 + 2^{1-L})\tilde{v}_{i\ell} \leq \tilde{v}_{i\ell} + 2^{1-L} v_{i\ell} \leq \tilde{v}_{i\ell} + U \cdot 2^{1-L}$} \\
            &\leq \tilde{v}_{i \ell'} \tilde{p}_{\ell} + (2 U \cdot 2^{-L} + 48k U \cdot 2^{-L} + 24k U \cdot 2^{-L} \cdot 2^{1-L}) - k^{-1} (U D)^{-k^2} \tag{by $\tilde{p}_\ell \leq 1$ $\forall\, \ell$ and $\tilde{v}_{i\ell'} \leq v_{i\ell} \leq U$} \\ 
            &\leq \tilde{v}_{i \ell'} \tilde{p}_{\ell} + 74k U \cdot 2^{-L} - k^{-1} (U D)^{-k^2} \tag{by $1 \leq k$ and $2^{1-L} \leq 1$} \\ 
            &\leq \tilde{v}_{i \ell'} \tilde{p}_{\ell} + 74 \cdot 2^{-8} k^{-1} (UD)^{-k^2} - k^{-1} (U D)^{-k^2} \tag{$D \geq 1$} \\ 
            &< \tilde{v}_{i \ell'} \tilde{p}_{\ell}, 
    \end{align}
    which contradicts $\arc{i\ell} \in \tilde{\mathcal{E}}(\tilde{\p})$. 
    In the same way, we can show that 
    $\arc{\ell j} \in \tilde{\mathcal{E}}(\tilde{\mathbf{p}})$, then $\arc{\ell j} \in \mathcal{E}(\mathbf{p}^*)$. 
\end{proof}

This completes the proof of~\cref{lem:find-an-eq-price-in-polytime}. 
\end{proof}


